\documentclass[11pt,reqno]{amsart}

\usepackage[utf8]{inputenc}
\usepackage[T1]{fontenc}
\usepackage{amsthm}
\usepackage{amsmath}
\usepackage{amssymb}
\usepackage[english]{babel}
\usepackage{bm}
\usepackage{comment}
\usepackage{cancel}
\usepackage[inline]{enumitem}
\usepackage[dvips]{graphicx}
\usepackage{color}
\usepackage{hyperref}

\usepackage{tikz}
\usetikzlibrary{matrix,arrows,positioning,automata,shapes,shadows,calc,fadings,decorations,snakes,through,intersections}

\tikzstyle{solid}=                   [dash pattern=]
\tikzstyle{dotted}=                  [dash pattern=on \pgflinewidth off 2pt]
\tikzstyle{densely dotted}=          [dash pattern=on \pgflinewidth off 1pt]
\tikzstyle{loosely dotted}=          [dash pattern=on \pgflinewidth off 4pt]
\tikzstyle{dashed}=                  [dash pattern=on 3pt off 3pt]
\tikzstyle{densely dashed}=          [dash pattern=on 3pt off 2pt]
\tikzstyle{loosely dashed}=          [dash pattern=on 3pt off 6pt]
\tikzstyle{dashdotted}=              [dash pattern=on 3pt off 2pt on \the\pgflinewidth off 2pt]
\tikzstyle{densely dashdotted}=      [dash pattern=on 3pt off 1pt on \the\pgflinewidth off 1pt]
\tikzstyle{loosely dashdotted}=      [dash pattern=on 3pt off 4pt on \the\pgflinewidth off 4pt]

\usepackage{url}
\usepackage{stackrel}
\usepackage[left=3cm, right=3cm, paperheight=12.8in]{geometry}  
\usepackage{fancyhdr}
\usepackage{mathrsfs}
\usepackage{stmaryrd}
\usepackage{soul}  
\usepackage{nicefrac}
\usepackage{datetime}
\longdate

\newtheorem{theorem}{Theorem}
\newtheorem{lemma}{Lemma}

\theoremstyle{definition}

\newtheorem{defi}{Definition}
\let\olddefi\defi
\renewcommand{\defi}{\olddefi\normalfont}
\newtheorem{example}{\textsc{Example}}
\newtheorem{remark}{\textbf{Remark}}

\theoremstyle{remark}
\newtheorem{claim}{\textsc{Claim}}
\newtheorem*{claim*}{\textsc{Claim}}

\hypersetup{
    pdftitle={Behavioral Equivalence of Extensive Game Structures},
    pdfmenubar=false,
    pdffitwindow=true,
    pdfstartview=FitH,
    colorlinks=true,
    linkcolor=blue,
    citecolor=green,
    urlcolor=cyan
}

\AtBeginDocument{%
   \def\MR#1{}
}

\providecommand{\MR}[1]{}

\providecommand{\MR}{\relax\ifhmode\unskip\space\fi MR }

\providecommand{\href}[2]{#2}

\DeclareMathSymbol{\widehatsym}{\mathord}{largesymbols}{"62}

\begin{document}
\title{Behavioral Equivalence of Extensive Game Structures}

\author{Pierpaolo Battigalli}
\address{Universit\`a ``Luigi Bocconi''---via Roentgen 1, 20136 Milano, Italy}
\email{pierpaolo.battigalli@unibocconi.it}

\author{Paolo Leonetti}
\address{Universit\`a ``Luigi Bocconi''---via Roentgen 1, 20136 Milano, Italy}
\email{paolo.leonetti@unibocconi.it}

\author{Fabio Maccheroni}
\address{Universit\`a ``Luigi Bocconi''---via Roentgen 1, 20136 Milano, Italy}
\email{fabio.maccheroni@unibocconi.it}


\makeatletter
\@namedef{subjclassname@1991}{JEL code}
\makeatother

%


\begin{abstract}
\noindent{} Two extensive game structures with imperfect information are
said to be behaviorally equivalent if they share the same map (up to
relabelings) from profiles of structurally reduced strategies to induced
terminal paths. We show that this is the case if and only if one can be
transformed into the other through a composition of two elementary
transformations, commonly known as \textquotedblleft Interchanging of
Simultaneous Moves\textquotedblright\ and \textquotedblleft Coalescing Moves/Sequential Agent Splitting.\textquotedblright 
\end{abstract}
\maketitle
\thispagestyle{empty}

\section{Introduction}\label{sec:introduction}

Fix a player, viz. player $i$, in the extensive-form representation of a
finite game. Say that two strategies $s_{i}^{\prime }$ and $s_{i}^{\prime
\prime }$ of player $i$ are \textbf{behaviorally equivalent} if they are
consistent with (do not prevent from being reached) the same collection of
information sets of player $i$ and behave in the same way at such
information sets. Two strategies $s_{i}^{\prime }$ and $s_{i}^{\prime \prime
}$ are realization-equivalent if, for every given strategy profile $%
s_{-i}$ of the opponents (including the chance player), $s_{i}^{\prime }$
and $s_{i}^{\prime \prime }$ induce the same play path, or terminal node 
(of course, the induced terminal node may depend on $s_{-i}$). Kuhn proved
that these two equivalence relations coincide, see \cite[Theorem 1]{MR0054924}. Intuitively, two strategies of player $i$ are equivalent in
this sense, if they respond in the same way to information about other
players, but may respond in different ways to past moves of $i$. This
corresponds to a notion of \textquotedblleft forward
planning,\textquotedblright\ such as planning the sequence of actions to be
taken in a one-person game (without chance moves). For this reason, some
authors call \textquotedblleft plan of action\textquotedblright\ such an
equivalence class of strategies, cf. \cite[Section 6.4]{MR1301776} and \cite%
{Rubinstein91}. 
Several authors studying the foundations of game theory
put forward solution concepts that do not distinguish between behaviorally
equivalent strategies. Such solution concepts rely on the following notion
of sequential best reply. For each of his 
information sets, player $i$ has a probabilistic belief about the strategies
of the opponents and the resulting system of beliefs satisfies the rules of
conditional probability whenever they apply. A strategy $s_{i}^{\star }$ is a 
\textbf{sequential best reply} to such system of beliefs if, for each
information set of $i$ reachable under $s_{i}^{\star }$, there is no
alternative strategy that yields a higher expected utility given the belief
at such information set. Different notions of \textquotedblleft
extensive-form rationalizability\textquotedblright\ put forward and applied
in the last four decades (e.g., Pearce \cite{Pearce}, Ben Porath \cite%
{BenPorath}, Battigalli \cite{Battigalli2003}, Battigalli and Siniscalchi 
\cite{MR2001808}) 
rely on this notion of sequential best reply and therefore do not
distinguish between behaviorally equivalent strategies. 
These solution concepts have been
rigorously justified by explicit and formal assumptions about strategic
reasoning. See, for example, the surveys of Battigalli and Bonanno \cite%
{BattigalliBonanno1999a}, and Dekel and Siniscalchi \cite{Dekel}, the
textbook of Perea \cite{MR3155151}, 
and the relevant references therein for rigorous and transparent epistemic
justifications of these solution concepts.\footnote{%
There are also extensive-form refinements of the Nash equilibrium concept,
such as Reny's explicable equilibrium \cite{MR1163001} that do not
distinguish between behaviorally equivalent strategies.}

Since behavioral and realization equivalence coincide, if two strategies are behaviorally equivalent, they are necessarily payoff equivalent; but it is easy to show by example that the converse is not true (see below). 
Thus, behavioral equivalence, sequential best reply, and the related solution concepts are not normal-form invariant. It is also clear from the definition of
behavioral equivalence and its coincidence with realization equivalence that
the only relevant part of a game in extensive form that matters to ascertain
such equivalence is the \textbf{game structure} (that is, the game tree with
the information structure), not the map from terminal nodes to consequences,
or payoffs. Call a class of behaviorally equivalent strategies a \textbf{%
structurally reduced strategy.} A transformation of a game structure
preserves behavioral equivalence if, up to relabeling, the new game
structure has the same set of structurally reduced strategies for each
player, the same set of terminal nodes, and the same map from profiles of
structurally reduced strategies to terminal nodes as the initial game
structure. In this case, we say that two game structures are \textbf{%
behaviorally equivalent. }

As anticipated, the notions of realization equivalence and behavioral equivalence are related to, but distinct from payoff equivalence of
strategies in a game, and normal-form equivalence of two games in extensive
form. The latter relations can be defined only for fully specified games whereby
each terminal node is associated with a corresponding profile of payoffs.
Two behaviorally equivalent strategies are necessarily \textbf{payoff
equivalent}, that is, independently of how the co-players' behavior is
fixed, they yield the same profile of payoffs. But the possibility of ties
between terminal nodes implies that the converse is not true. The standard
notion of equivalence between games requires the same map (up to relabeling)
between profiles of classes of payoff equivalent strategies and payoff
profiles. If we assign the same payoff functions (up to relabeling of
terminal nodes) to two behaviorally equivalent game structures, we obtain
two games in extensive form that are equivalent in the (less demanding)
traditional sense, that is, they have the same reduced normal form. But such
traditional equivalence may also hold between non behaviorally equivalent
games. For example, a (non trivial) $2$-person game with perfect
information and a game with (essentially) simultaneous moves with the same
normal form are not behaviorally equivalent.

Since there are well-known and relevant solution concepts that are not
normal-form invariant, but yield the same solutions (up to relabeling) for
games with the same structure and the same payoff functions, we find it
interesting to provide a characterization of behavioral equivalence by means
of invariant transformations. We consider a class of extensive-form
structures allowing a direct representation of simultaneous moves whereby
more than one player may be active at a given node; see, e.g., \cite[Section
6.3.2]{MR1301776}. Therefore, besides the standard transformation of
\textquotedblleft Interchanging\textquotedblright\ essentially simultaneous,
but formally sequential moves, we apply a \textquotedblleft
Simultanizing\textquotedblright\ transformation that makes such moves truly
simultaneous. We also consider the standard transformation of
\textquotedblleft Coalescing Moves\textquotedblright\, and its
inverse, that is, \textquotedblleft Sequential Agent Splitting\textquotedblright\ (see the examples below). With this, we prove the following result: 
\emph{Two finite game structures }$G$ \emph{and }$G^{\prime}$ 
\emph{are behaviorally equivalent if and only if there is a} 
\emph{sequence of transformations of the \textquotedblleft
Interchanging/Simultanizing\textquotedblright\ kind and \textquotedblleft
Coalescing Moves\textquotedblright\ kind and their inverses that connects }$%
G$ \emph{with }$G^{\prime}$\emph{.}

\begin{example}
\label{ex:coalescingexample} Consider player \textsl{1} in the extensive game structures $G$ and $G^{\prime }$ represented in Figure \ref{fig:coa1}.
Neither the consequence function nor the preference relations over the set
of terminal nodes $\{z_{1},z_{2},z_{3}\}$ are specified. In particular, the
\textquotedblleft normal-form\textquotedblright\ representations will be
defined as maps from strategy profiles to terminal nodes. It has been
argued that player \textsl{1} should regard $G$ and $G^{\prime }$ as
representing \textquotedblleft essentially\textquotedblright\ the same
situation, if his preferences over $\{z_{1},z_{2},z_{3}\}$ are the same. In
particular, strategies $x.a\,$\ and $x.b$ are behaviorally equivalent and the
set $\left\{ x.a,x.b\right\} $ gives the structurally reduced strategy $x$,
representing the \textquotedblleft plan\textquotedblright\ of reaching 
$z_{1}$, whereas $y.a$ (resp. $y.b$) corresponds to the plan of reaching $%
z_{2}$ (resp. $z_{3}$). Thus, $G$ and $G^{\prime }$ have isomorphic sets (of
cardinality 3) of structurally reduced strategies, where the isomorphism
preserves the map from reduced strategies to terminal nodes.

\begin{figure}[!htb]
\centering
\begin{minipage}[c]{0.40\textwidth}
\centering
\begin{tikzpicture}
[scale=1,description/.style=auto]
\node (a01) at (0,3){{\textsl{1}}};
\node (a11) at (-2,1.5){{\footnotesize $z_1$}};
\node (a12) at (2,1.5){{\textsl{1}}};
\node (a21) at (0.5,0){{\footnotesize $z_2$}};
\node (a22) at (3.5,0){{\footnotesize $z_3$}};
\draw [-latex,thin](a01)--(a11);
\draw [-latex,thin](a01)--(a12);
\draw [-latex,thin](a12)--(a21);
\draw [-latex,thin](a12)--(a22);
\node (label01) at (-1.3,2.4){{\footnotesize $x$}};
\node (label02) at (1.3,2.4){{\footnotesize $y$}};
\node (label03) at (1.05,.9){{\footnotesize $a$}};
\node (label04) at (2.95,.9){{\footnotesize $b$}};
\end{tikzpicture}
\end{minipage}
\hspace{11mm} 
\begin{minipage}[c]{0.40\textwidth}
\centering
\begin{tikzpicture}
[scale=1,description/.style=auto]
\node (a01) at (0,3){{\textsl{1}}};
\node (a11) at (-2,1.5){{\footnotesize $z_1$}};
\node (a21) at (0.5,0){{\footnotesize $z_2$}};
\node (a22) at (3.5,0){{\footnotesize $z_3$}};
\draw [-latex,thin](a01)--(a11);
\draw [-latex,thin](a01)--(a21);
\draw [-latex,thin](a01)--(a22);
\node (label01) at (-1.3,2.4){{\footnotesize $x$}};
\node (label03) at (0,1.5){{\footnotesize $a$}};
\node (label04) at (2.15,1.5){{\footnotesize $b$}};
\end{tikzpicture}
\end{minipage}
\par
\vspace{2mm}
\caption{Example of Coalescing Moves.}
\label{fig:coa1}
\end{figure}
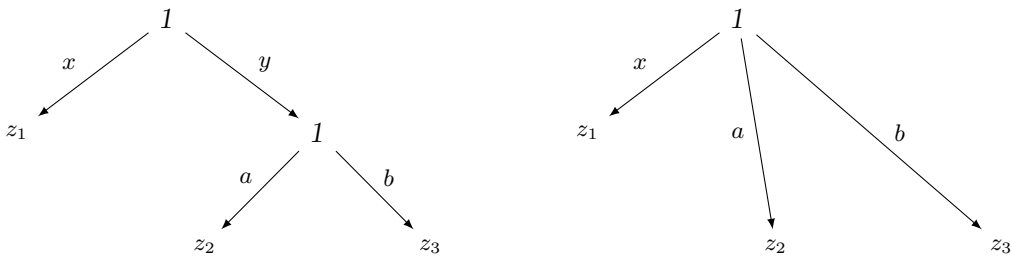
\end{example}

\begin{example}%
\label{ex:splittingxample} The two game structures in Figure \ref{fig:infl}
feature imperfect information. In both cases players \textsl{1} and \textsl{2%
} move sequentially and the second mover cannot observe the move of the
first, which is represented by joining with a dashed line the nodes where the
second mover is active, but the order of moves is interchanged.

\begin{figure}[!htb]
\centering
\begin{minipage}[c]{0.451\textwidth}
\centering
\begin{tikzpicture}
[scale=1,description/.style=auto]
\node (a01) at (0,3){{\textsl{1}}}; 
\node (a11) at (-2+.5,1.5){{\textsl{2}}};
\node (a12) at (2-.5,1.5){{\textsl{2}}};
\node (a21) at (1.5-.75,0){{\footnotesize $z_3$}};
\node (a22) at (1.5+.75,0){{\footnotesize $z_4$}};
\node (a31) at (-1.5-.75,0){{\footnotesize $z_1$}};
\node (a32) at (-1.5+.75,0){{\footnotesize $z_2$}};
\draw [-latex,thin](a01)--(a11);
\draw [-latex,thin](a01)--(a12);
\draw [-latex,thin](a12)--(a21);
\draw [-latex,thin](a12)--(a22);
\draw [-latex,thin](a11)--(a31);
\draw [-latex,thin](a11)--(a32);
\draw [dashed,thin](a11)--(a12);
\node (label01) at (-1.3+.3+.7-.65,2.4){{\footnotesize $x$}};
\node (label02) at (.9,2.4){{\footnotesize $y$}};
\node (label03) at (.9,.9){{\footnotesize $a$}};
\node (label04) at (2.1,.9){{\footnotesize $b$}};
\node (label05) at (.9-3,.9){{\footnotesize $a$}};
\node (label06) at (2.1-3,.9){{\footnotesize $b$}};
\end{tikzpicture}
\end{minipage}
\hspace{10mm} 
\begin{minipage}[c]{0.45\textwidth}
\centering
\begin{tikzpicture}
[scale=1,description/.style=auto]
\node (a01) at (0,3){{\textsl{2}}};
\node (a11) at (-2+.5,1.5){{\textsl{1}}};
\node (a12) at (2-.5,1.5){{\textsl{1}}};
\node (a21) at (1.5-.75,0){{\footnotesize $z_2$}};
\node (a22) at (1.5+.75,0){{\footnotesize $z_4$}};
\node (a31) at (-1.5-.75,0){{\footnotesize $z_1$}};
\node (a32) at (-1.5+.75,0){{\footnotesize $z_3$}};
\draw [-latex,thin](a01)--(a11);
\draw [-latex,thin](a01)--(a12);
\draw [-latex,thin](a12)--(a21);
\draw [-latex,thin](a12)--(a22);
\draw [-latex,thin](a11)--(a31);
\draw [-latex,thin](a11)--(a32);
\draw [dashed,thin](a11)--(a12);
\node (label01) at (-1.3+.3+.7-.65,2.4){{\footnotesize $a$}};
\node (label02) at (.9,2.4){{\footnotesize $b$}};
\node (label03) at (.9,.9){{\footnotesize $x$}};
\node (label04) at (2.1,.9){{\footnotesize $y$}};
\node (label05) at (.9-3,.9){{\footnotesize $x$}};
\node (label06) at (2.1-3,.9){{\footnotesize $y$}};
\end{tikzpicture}
\end{minipage}
\par
\vspace{2mm}
\par
\caption{Example of Interchanging/Simultanizing.}
\label{fig:infl}
\end{figure}
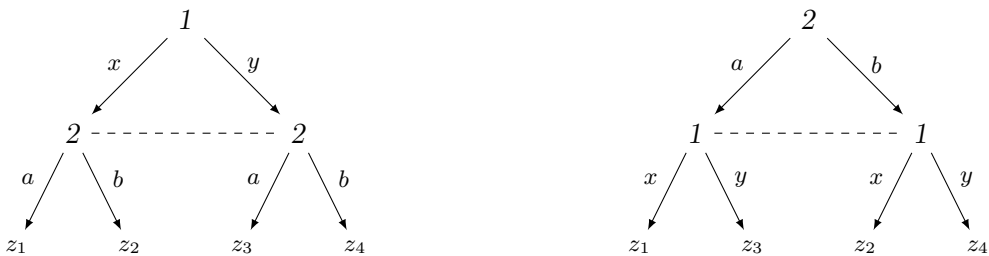

It has been argued that these two structures represent \textquotedblleft
essentially\textquotedblright\ the same situation. 
\end{example}


\subsection{Literature}

\label{subsec:review} The standard literature on game equivalence aims at
characterizing classes of games with the same reduced normal form by means
of invariant transformations. The first work on game equivalence is Thompson 
\cite{Thompson}, who defines four transformations commonly known as
\textquotedblleft Interchanging of Simultaneous Moves,\textquotedblright\
\textquotedblleft Coalescing Moves/Sequential Agent
Splitting,\textquotedblright\ \textquotedblleft Addition of a Superfluous
Move,\textquotedblright\ and \textquotedblleft Inflation/Deflation.\textquotedblright\ Relying on the simplification of Krentel,
McKinsey, and Quine \cite{Krentel} and the extensive-game model proposed by Kuhn 
\cite{Kuhn}, he shows that, up to relabelings, two finite games
share the same reduced normal form if and only if each extensive form
representation can be transformed into the other through a finite number of
applications of these transformations (see also \cite{Dalkey}).

A few later contributions extended Thompson's work. 
In particular, Kohlberg and Mertens \cite{KohlbergMertens} extend the above
result to games with chance moves, proposing two additional transformations
which are, essentially, modified versions of Coalescing Moves and Addition
of Superfluous Moves for the chance player. They argue that all the
\textquotedblleft strategic features\textquotedblright\ are unchanged
through the application of these transformations. 
Elmes and Reny \cite{ElmesReny} point out that one of these transformations,
Inflation/Deflation, does not preserve the perfect recall property. Hence,
given a modified version of the Addition of a Superfluous Move
transformation, they show that two extensive form games with the same
reduced normal form can be transformed into each other without using 
the unwanted transformation, and preserving the perfect recall property. This
is, in our view, conceptually appealing, because whether players have
perfect or imperfect recall should be part of the description of their
personal traits (like their subjective preferences), not of the objective rules of
the game. Such rules should describe the objective \emph{\textquotedblleft
flow\textquotedblright } of information to players. 
Assuming that players are necessarily informed of the actions they just
chose, the objectively accumulated \emph{\textquotedblleft
stock\textquotedblright } of information is represented by information
partitions with the perfect recall property, see also Section \ref{subsec:perectrecall}.   
Other notions of game equivalence have been studied in the literature. For example, Hoshi and Isaac \cite{HoshiIsaac} extend Thompson's result to the case of games with unawareness. 
Bonanno \cite{Bonanno} investigates the notion of game equivalence which arises from the iterated application of Interchanging of Simultaneous Moves only, showing that the resulting equivalence classes correspond to set-theoretic reductions. 
Goranko \cite{MR2037178} takes a linguistic approach, gives an axiomatization of the algebra of games and defines canonical forms so that every game term is provably equivalent to a minimal canonical one. 
Lastly, van Benthem, Bezhanishvili, and Enqvist \cite{vanBenthem2018} look at equivalence from the perspective of the players’ power for controlling outcomes.

Compared to most of this literature, we differ in two ways. First, we consider a different set of transformations, which just apply to the extensive form structure, irrespective of the outcome function.\footnote{In this respect, our work is similar to Bonanno \cite{Bonanno}.} 
Second, our formal language is somewhat different. In particular, for us actions, rather than nodes, are the primitive terms; furthermore, we provide a direct representation of simultaneity. The latter is in principle important, because two games with different sequences of \textquotedblleft essentially simultaneous\textquotedblright\, moves are not necessarily viewed as equivalent by the players. For example, only the first mover can be afraid (or hopeful) of being spied on. More generally, we can cleanly represent incomplete information about the order of moves and the information structure as in \cite{Penta2019}, whereas the representation à la Kuhn \cite{MR0054924} would be cumbersome.


\section{Preliminaries on Game structures}\label{sec:framework}

\subsection{Preliminary Notation}\label{sec:notation}
Let $(X,\le)$ be a finite partially ordered set, i.e., a nonempty finite set $X$ with a binary relation $\le$ contained in $X\times X$ which is transitive, reflexive, and antisymmetric. As usual, given $x,y \in X$, we write $x<y$ as an abbreviation of $x\le y$ and $x\neq y$. 
We let $[x,y]$ represent the order interval $\{v \in X: x\le v\le y\}$ (hence $[x,y]\neq \emptyset$ if and only if $x\le y$) and denote the immediate predecessor relation by $\ll$, that is, $x\ll y$ with $x \neq y$ if and only if $[x,y]=\{x,y\}$. 
A finite partially ordered set $(X,\le)$ is said to be a \textbf{tree} whenever there exists a (necessarily unique) minimum element $e$, which is called the root of the tree, and for each $x \in X$ the order interval $[e,x]$ is totally ordered. 

Given a nonempty set $S$, we denote by $2^S$ the power set of $S$ and by $(\bm{\mathscr{P}}\text{{\scriptsize art}}(S), \le)$ the collection of partitions of $S$. 
The latter is partially ordered by the refinement relation, that is, $\mathscr{P} \le \mathscr{P}^\prime$ for some partitions $\mathscr{P},\mathscr{P}^\prime$ of $S$ 
whenever, for each $P \in \mathscr{P}$, there exists a sub-collection $\{P_i^\prime: i \in I\}\subseteq \mathscr{P}^\prime$ such that $P=\bigcup_{i \in I}P_i^\prime$. 
In addition, we represent the set of finite sequences of elements from $S$ by 
$$
\textstyle S^{<\mathbb{N}_0}:=\bigcup_{n \in \mathbb{N}_0}S^n,
$$
where $S^0:=\{\varnothing\}$ is the singleton containing the empty sequence (here $\mathbb{N}_0$ and $\mathbb{N}$ stand, respectively, for the set of nonnegative integers and positive integers; in particular $0 \in \mathbb{N}_0$).  
Accordingly, we define a partial order $\preceq$ on $S^{< \mathbb{N}_0}$ such that $x\preceq y$ if and only if $x=(x_1,\ldots,x_n) \in S^n$ is a prefix of $y=(y_1,\ldots,y_m) \in S^m$, that is, if and only if $x=\varnothing$ or $0<n\le m$ and $x_i=y_i$ for all positive integers $i\le n$. 
Note that $\varnothing \preceq x$ for all $x \in S^{< \mathbb{N}_0}$, every nonempty order interval $[x,y]$ is finite, hence totally ordered, and the set $\{v: x\prec v \preceq y\}$ has a minimum provided that $x\prec y$; cf. \cite[pp. 144-145]{MR3497869}. 
Thus $(S^{< \mathbb{N}_0}, \preceq)$ is a tree. 
Here, the immediate predecessor relation is denoted by $\prec \hspace{-1.5mm}\prec$.

Recall that a complete lattice is a partially ordered set such that all subsets have both a supremum and an infimum. In particular, a complete lattice admits a greatest and a least element.  Accordingly, we have the following result (we omit details).
\begin{lemma}\label{partitions}
Fix a finite set $X$. Then the partially ordered set $(\bm{\mathscr{P}}\text{{\scriptsize art}}(X), \le)$ is a complete lattice. 
In particular, for each nonempty collection $\{\mathscr{P}_j: j \in J\}$ of partitions of $X$, $\sup\left\{\mathscr{P}_j: j \in J\right\}$ exists in $\bm{\mathscr{P}}\text{{\scriptsize art}}(X)$.
\end{lemma}

Lastly, given nonempty sets $X,Y,Z$, a correspondence $\varphi: X \twoheadrightarrow Y$ is a map that associates each $x \in X$ with some \emph{possibly empty} subset of $Y$. 
With this, for each $y\in Y$, 
we define the inverse image of $y$ as 
$\varphi^{-1}(y):=\{x \in X: y\in \varphi(x)\}$ 
(hence, $\varphi^{-1}$ is a correspondence from $Y$ to $X$). 
Given another correspondence $\psi: Y \twoheadrightarrow Z$, the composition $\psi\circ \varphi: X\twoheadrightarrow Z$ is the correspondence defined by $(\psi \circ \varphi)(x):=\bigcup_{y \in \varphi(x)}{\psi(y)}$ for all $x \in X$. 
Note that correspondences can be equivalently reinterpreted as binary relations (see also the comments in Section \ref{sec:mainth}).


\subsection{Extensive game structures} Let $\mathcal{G}$ be the collection of all extensive game structures $G$ with imperfect information and perfect recall represented by tuples
\begin{equation}\label{eq:extensive}
\langle I, \bar{H}, (A_i, \emph{\textbf{H}}_i)_{i \in I} \rangle,
\end{equation}
where each primitive component will be described below (this is inspired by \cite{MR1301776}). 
For brevity, we will often write \textquotedblleft game\textquotedblright\, instead of 	\textquotedblleft game structure.\textquotedblright\,  
\emph{Hereafter, all extensive game structures are assumed to be finite and without chance moves.} 

\subsubsection{Players} 
$I$ is a nonempty finite \textbf{set of players}. 

\subsubsection{Actions} 
For each $i \in I$, $A_i$ is a nonempty finite set of potentially \textbf{feasible actions} which can be chosen by player $i$. 
Then, we let
$$
A:=\bigcup_{\emptyset\neq J\subseteq I}\left(\prod_{i \in J}A_i\right)
$$
denote the set of action profiles of subsets of players. 

\subsubsection{Histories} 
$\bar{H}$ is the finite \textbf{set of histories}, i.e., a finite tree contained in $(A^{< \mathbb{N}_0},\preceq)$ such that $\varnothing \in \bar{H}$ and closed under the immediate predecessor relation, that is, if $x\prec \hspace{-1.5mm}\prec y$ in $(A^{< \mathbb{N}_0},\preceq)$ and $y \in \bar{H}$, then $x \in \bar{H}$. We refer to elements of $\bar{H}$ as \textquotedblleft histories\textquotedblright\, or \textquotedblleft nodes\textquotedblright\,as convenient. (Hence, each node in $\bar{H}$ is a chain of profile of actions in $A$.) 
With this, $\bar{H}$ can be partitioned in the sets of terminal histories $Z$ and non-terminal histories $H:=\bar{H}\setminus Z$. 

\medskip


From the previous elements, we derive an \textbf{active-player correspondence} $I(\cdot): H \twoheadrightarrow I$ that associates with each $h \in H$ the \emph{nonempty} subset $I(h) \subseteq I$ such that
$$
\forall a \in A,\quad (h,a) \in \bar{H} \quad \implies \quad a \in \prod_{i \in I(h)}A_i.
$$
In addition, we assume that, for each $i \in I$, there is some $h \in H$ such that $i \in I(h)$, that is,
$$
H_i:=\left\{h \in H: i \in I(h)\right\}\neq \emptyset.
$$

\subsubsection{Information structure} For each player $i \in I$, $\emph{\textbf{H}}_i$ is a partition of $H_i$, called the \textbf{information partition} of $i$.

\medskip

Moreover, we assume that the profile $(F_{i}:H\twoheadrightarrow A_{i})_{i\in I}$ of \textbf{feasibility
correspondences} defined, for each $i \in I$, by
$$
\textstyle \forall h \in H,\quad F_i(h):=\mathrm{proj}_{A_{i}}\left\{ a\in \prod_{j\in I\left(h\right) }A_{j}:\left( h,a\right) \in \bar{H}\,\right\} 
$$
satisfies the following properties:
\begin{list}{$\diamond$}{}
\item \label{oldassumptionA} For each player $i \in I$ and each non-terminal history $h \in H$, if $F_i(h)\neq \emptyset$, then $|F_i(h)|\ge 2$; in the latter case, we say that $i$ is active, otherwise he\footnote{The male pronouns (he, him, his) are used throughout this paper. We hope this won't be interpreted by anyone as an attempt to exclude females from the game or to imply their exclusion. Centuries of use have made these pronouns neutral, and we feel their use provides for clear and concise written text.} is inactive.
\item\label{item:G6}
For each non-terminal history $h \in H$, the feasibility conditions of distinct players are logically independent, that is, 
$$
\textstyle \forall a \in \prod_{i \in I(h)}A_i,\quad h\prec \hspace{-1.5mm}\prec (h,a)\,\,\,\Longleftrightarrow\,\,\, a\in \prod_{i \in I(h)}F_i(h).
$$
Hence, we write $(h,a):=(a_1,\ldots,a_n,a)$ whenever $h=(a_1,\ldots,a_n) \in H$; in particular, $(h,a)=(a)$ if $h=\varnothing$.
\item For each player $i \in I$, $F_i$ is $\emph{\textbf{H}}_i$-measurable, that is, 
$$
\textstyle \forall \emph{\textbf{h}}_i \in \emph{\textbf{H}}_i, \forall h,h^\prime \in \emph{\textbf{h}}_i,\quad F_i(h)=F_i(h^\prime). 
$$
Each element $\emph{\textbf{h}}_i \in \emph{\textbf{H}}_i$ is interpreted as an \textbf{information set} of player $i$, i.e., 
a maximal subset of non-terminal histories that player $i$ is not able to distinguish. 
With this, it makes sense to write $F_i(\emph{\textbf{d}}_i)$ for each nonempty subset $\emph{\textbf{d}}_i\subseteq \emph{\textbf{h}}_i$. 
Note that, at an information set, player $i$ might not know who his active 
opponents are and their feasible actions.
\item For notational convenience, we assume that different information sets are associated with disjoint sets of feasible actions, that is, 
\begin{equation}\label{eq:disjointinfosets}
\forall \emph{\textbf{h}}_i, \emph{\textbf{h}}_i^\prime \in \emph{\textbf{H}}_i, \quad  \emph{\textbf{h}}_i\neq \emph{\textbf{h}}_i^\prime \,\,\implies\,\, F_i(\emph{\textbf{h}}_i)\cap F_i(\emph{\textbf{h}}_i^\prime)=\emptyset. 
\end{equation}
\end{list}

%

Lastly, for each history $h \in \bar{H}$, let $Z(h):=\{z \in Z: h\preceq z\}$ be the set of terminal histories which follow $h$  and, for each nonempty $U\subseteq \bar{H}$, write $Z(U):=\bigcup_{h \in U}Z(h)$. Also, for each information set $\emph{\textbf{h}}_i \in \emph{\textbf{H}}_i$ and feasible action $a_i^\star \in F_i(\emph{\textbf{h}}_i)$, define the set of terminal histories which follow $\emph{\textbf{h}}_i$ and $a_i^\star$ as
$$
Z(\emph{\textbf{h}}_i,a_i^\star):=\bigcup_{h \in \emph{\textbf{h}}_i} \bigcup_{a_{-i} \in F_{-i}(h)} Z((h,(a_i^\star,a_{-i}))).
$$

The following example illustrates and explains our formalism. 

\begin{example}
Consider the extensive game structure represented in Figure \ref{fig:examplegamestructure}.


\begin{figure}[!htbp]
\centering
\begin{tikzpicture}
[xscale=1.2,description/.style=auto]
\node (00) at (0,0){\textsl{1,2}};
\node (22) at (2,2){{\small $z_2$}};
\node (2-2) at (2,-2){{\small $z_3$}};
\node (-22) at (-2,2){{\small $z_1$}};
\node (-2-2) at (-2,-2){\textsl{3}};
\node (a) at (-4,-4){{\small $z_4$}};
\node (b) at (-1.5,-4){{\small $z_5$}};
\draw [-latex,thin] (00)--(22);
\draw [-latex,thin] (00)--(2-2);
\draw [-latex,thin] (00)--(-22);
\draw [-latex,thin] (00)--(-2-2);
\draw [-latex,thin] (-2-2)--(a);
\draw [-latex,thin] (-2-2)--(b);
\node (22_label) at (1.6,1.1){{\small $(u,r)$}};
\node (2-2_label) at (1.6-.1,-1){{\small $(d,r)$}};
\node (-22_label) at (-1.6,1.1){{\small $(u,\ell)$}};
\node (-2-2_label) at (-1.6+.1,-1){{\small $(d,\ell)$}};
\node (a_label) at (-3.5+.08,-3.1){{\small $x$}};
\node (b_label) at (-1.5,-3.1){{\small $y$}};
\end{tikzpicture}
\caption{An example of extensive game structure.}
\label{fig:examplegamestructure}
\end{figure}
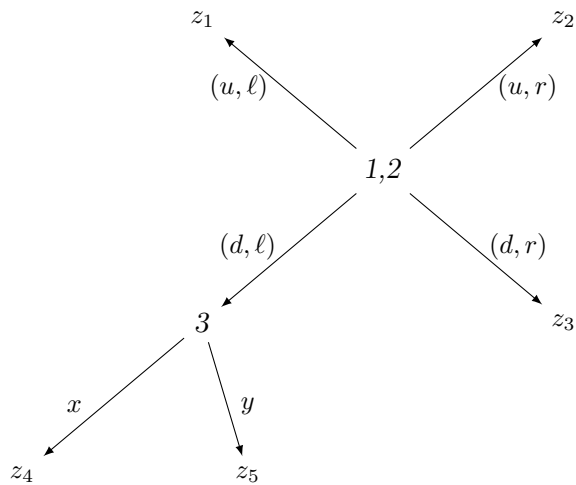


Here, the set of players is $I=\{1,2,3\}$ and the set of feasible actions are $A_1=\{u,d\}$, $A_2=\{\ell,r\}$, and $A_3=\{x,y\}$. The set of histories $\bar{H}$ is partitioned into the set of terminal histories $Z=\{((u,\ell)), ((u,r)), ((d,r)), ((d,\ell),x), ((d,\ell),y)\}$ and non-terminal histories $H=\{\varnothing, h\}$, where $h:=((d,\ell))$. In addition, the feasibility correspondences satisfy $F_i(\varnothing)=A_i$ and $F_i(h)=\emptyset$ for each $i \in \{1,2\}$, $F_3(\varnothing)=\emptyset$, and $F_3(h)=A_3$. Lastly, all informations sets are singletons, so that $\emph{\textbf{H}}_i=\{\{\varnothing\}\}$ for each $i \in \{1,2\}$, and $\emph{\textbf{H}}_3=\{\{h\}\}$. Despite this, the game does not feature perfect information in the usual sense because it contains simultaneous moves by player \textsl{1} and \textsl{2} at the root.
\end{example}

\subsection{Perfect recall}\label{subsec:perectrecall} Game structures $G \in \mathcal{G}$ have to satisfy perfect recall, see \cite[Definition 6.5]{MR3497869}. This means that each player always remembers everything he knew and did earlier. 
Thus information sets represent the \textquotedblleft stock\textquotedblright\, of information objectively given to a player irrespective of his cognitive abilities, which are personal features and not part of the rules of the game.\footnote{We could instead describe explicitly the information flows implied by the rules of the game (as, for example, in Myerson \cite{MR832762}), but this would make notation heavier and proofs more convoluted.} 
In particular, perfect recall implies that, in each information set, two distinct histories are unrelated, i.e.,
\begin{equation}\label{eq:absmind}
\forall i \in I, \forall \emph{\textbf{h}}_i \in \emph{\textbf{H}}_i, \forall h,h^\prime \in \emph{\textbf{h}}_i,\,\,\,\,\,\,h\preceq h^\prime \,\,\,\implies\,\,\, h=h^\prime. 
\end{equation}
The violation of \eqref{eq:absmind} is commonly known as ``absent-mindedness'' (see \cite{MR1467534}). 


\subsection{Z-reduced normal form}\label{subsec:zred} For each player $i \in I$, let $S_i$ be his set of strategies, that is, 
$$
S_i:=\prod_{\emph{\textbf{h}}_i \in \emph{\textbf{H}}_i}F_i(\emph{\textbf{h}}_i).
$$
Similarly, 
$
S:=\prod_{i \in I}S_i 
$ 
denotes the set of strategy profiles. 
Each strategy profile $s \in S$ determines a unique terminal history $z \in Z$.\footnote{Players who are active at $\varnothing$ determine a unique history $h_1 \in \overline{H}$ of length $1$; if $h_1 \in Z$ we are done, otherwise players who are active at $h_1$ determine a unique history $h_2 \in \overline{H}$ of length $2$; this is repeated only a finite number of times, hence the algorithm terminates.} We denote this path function by 
$$
\zeta: S \to Z.
$$

Accordingly, for each game structure $G =\langle I, \bar{H}, (A_i, \emph{\textbf{H}}_i)_{i \in I} \rangle \in \mathcal{G}$, we define its $Z$\textbf{-normal form} by
\begin{displaymath}
\mathrm{n}_Z(G):=\langle I, (S_i)_{i \in I}, Z, \zeta \rangle.
\end{displaymath}
Note that the $Z$-normal form is not graphical per se, but rather represents the game $G$ by means of a kind of matrix.

\begin{defi}\label{defi:strategyequivalent}
Fix $G \in \mathcal{G}$ as in \eqref{eq:extensive} and a player $i \in I$. 
Two strategies $s_i,s_i^\prime \in S_i$ are \textbf{behaviorally equivalent}, written as $s_i \overset{i}{\sim} s_i^\prime$, if 
\begin{displaymath}
\forall s_{-i} \in S_{-i}, \,\,\,\,\zeta(s_i,s_{-i})=\zeta(s_i^\prime,s_{-i}),
\end{displaymath}
where $S_{-i}:=\prod_{j \in I\setminus \{i\}}S_j$. 
\end{defi}
It is worth noting that this is the classical definition of realization-equivalent strategies. However, as proved by Kuhn in \cite[Theorem 1]{MR0054924}, this coincides with the condition that strategies $s_i,s_i^\prime \in S_i$ do not prevent from being reached the same collection of information sets in $\emph{\textbf{H}}_i$ and behave in the same way at such information sets.

Since $\overset{i}{\sim}$ is an equivalence relation on $S_i$, we can define the quotient space 
$$\mathcal{S}_i:=S_i/\overset{i}{\sim}.$$
Elements of $\mathcal{S}_i$ are \textbf{structurally reduced strategies} of player $i$, that is, classes of behaviorally equivalent strategies of $i$. 
Similarly, we set $\mathcal{S}:=\prod_{i \in I}\mathcal{S}_i$ and denote the representative of each $s \in S$ by $s^\bullet \in \mathcal{S}$.

For each $G \in \mathcal{G}$,  
we let its $Z$\textbf{-reduced normal form} be
\begin{displaymath}
\mathrm{rn}_Z(G):=\langle I, (\mathcal{S}_i)_{i \in I}, Z, \tilde{\zeta} \rangle,
\end{displaymath}
where $\tilde{\zeta}: \mathcal{S} \to Z$ is the map defined by $\tilde{\zeta}(s^\bullet)=\zeta(s)$.

With this, we have all the ingredients to define the notion of behavioral equivalence for game structures.  
\begin{defi}\label{defi:behavioralequivalence}
Two extensive game structures $G,G^\prime \in \mathcal{G}$ are \textbf{behaviorally equivalent} 
if they share the same $Z$-reduced normal form up to isomorphisms\footnote{More explicitly, there exist a bijection $f_i: \mathcal{S}_i\to \mathcal{S}_i^\prime$ for each $i \in I$ and a bijection $g: Z \to Z^\prime$ such that $g(\tilde{\zeta}(s))=\tilde{\zeta^\prime}(f(s))$ for all $s \in \mathcal{S}$; here $f(s)$ is the structurally reduced strategy profile $(f_i(s_i))_{i \in I} \in \mathcal{S}^\prime$.} (i.e., $\mathrm{rn}_Z(G) \simeq \mathrm{rn}_Z(G^\prime)$).
\end{defi}
In particular, the two game structures in Example \ref{ex:coalescingexample} and the two game structures in Example \ref{ex:splittingxample} are, respectively, behaviorally equivalent.


\section{Invariant Transformations}\label{sec:invariant}

We introduce the notion of invariant transformation:
\begin{defi}\label{def:invarianttransformation}
A correspondence $T: \mathcal{G} \twoheadrightarrow \mathcal{G}$ is said to be an \textbf{invariant transformation} if $\mathrm{dom}(T):=\{G \in \mathcal{G}: T(G)\neq \emptyset\}\neq \emptyset$ and $G$ is behaviorally equivalent to $G^\prime$ for all $G \in \mathrm{dom}(T)$ and $G^\prime \in T(G)$. The family of invariant transformations will be denoted by $\mathcal{T}$. 
\end{defi}



To define the basic invariant transformations, we first need to introduce the notions of \textquotedblleft controlling\textquotedblright\, and \textquotedblleft dominating\textquotedblright\, sets. 
Intuitively, these sets are collections of histories in the same information set of an active player $i$ that, in a sense, control suitable paths which are successors of a given set of histories $h \in H$. 


\subsection{Controlling and dominating sets} 
Fix 
a game structure 
$G \in \mathcal{G}$ as defined in \eqref{eq:extensive}. 

\medskip

\begin{defi}\label{def:controllingset}
Given a player $i \in I$ and two information sets $\emph{\textbf{h}}_i, \emph{\textbf{h}}_i^\prime \in \emph{\textbf{H}}_i$, we say that $\emph{\textbf{h}}_i^\prime$ \textbf{controls} $\emph{\textbf{h}}_i$, written 
$$
\emph{\textbf{h}}_i \ll_i \emph{\textbf{h}}_i^\prime\,,
$$
whenever there exists a feasible action $a_i^\star \in F_i(\emph{\textbf{h}}_i)$ such that 
\begin{equation}\label{eq:controlling}
Z(\emph{\textbf{h}}_i,a_i^\star)=Z(\emph{\textbf{h}}_i^\prime).
\end{equation}
\end{defi}
In other words, given perfect recall, $\emph{\textbf{h}}_i \ll_i \emph{\textbf{h}}_i^\prime$ 
means that $\emph{\textbf{h}}_i^\prime$ follows $\emph{\textbf{h}}_i$ in the directed forest of information sets of player $i$,\footnote{That is, $\emph{\textbf{h}}_i$ precedes $\emph{\textbf{h}}_i^\prime$ in the directed forest if, for each $h^\prime \in \emph{\textbf{h}}_i^\prime$ there is $h \in \emph{\textbf{h}}_i$ such that $h \prec h^\prime$.} and the only thing $i$ learns moving from $\emph{\textbf{h}}_i$ to $\emph{\textbf{h}}_i^\prime$ is that he chose at $\emph{\textbf{h}}_i$ the unique action leading to $\emph{\textbf{h}}_i^\prime$ (hence, $\emph{\textbf{h}}_i^\prime$ is an immediate follower of $\emph{\textbf{h}}_i$ in the directed forest, given that player $i$ has at least two actions at each information set).

\bigskip

For example, the bottom information set of player \textsl{2} in the left hand side of Figure \ref{fig:coalescinggeneral} controls the top information set.

\begin{remark}\label{remark:controllingset}
Fix $i \in I$. Given an information set $\emph{\textbf{h}}_i \in \emph{\textbf{H}}_i$ and a feasible action $a_i^\star \in F_i(\emph{\textbf{h}}_i)$, it is easily seen, due to perfect recall, that there is at most one $\emph{\textbf{h}}_i^\prime \in \emph{\textbf{H}}_i$ following action $a_i^\star$ such that $\emph{\textbf{h}}_i \ll_i \emph{\textbf{h}}_i^\prime$. 
\end{remark}

\medskip

\begin{defi}\label{def:dominatingset}
Given a player $i \in I$, a non-terminal history $h \in H$ with $i\notin I(h)$, and a nonempty subset $\emph{\textbf{d}}_i$ of some information set $\emph{\textbf{h}}_i \in \emph{\textbf{H}}_i$, 
we say that $\emph{\textbf{d}}_i$ \textbf{dominates} $h$, written 
$$
h \lessdot_i \emph{\textbf{d}}_i,
$$
whenever
$$
Z(h)=Z(\emph{\textbf{d}}_i).
$$
\end{defi}
In other words, $h$ preceeds $\emph{\textbf{d}}_i$ and they have the same set of terminal successors. 
For example, the two nodes on the right of the information set of player \textsl{3} in the first game tree of Figure  \ref{fig:interchanginggeneral} dominate history $(r)$, 
i.e., $(r) \lessdot_{\textsl{3}} \{(r,a),(r,b)\}$.

\begin{remark}\label{remark:dominatingset} 
In the same spirit of Remark \ref{remark:controllingset}, fix $h \in H$ and $i \in I$ such that $i$ is inactive at $h$. Then it is easily seen that there is at most one nonempty set $\emph{\textbf{d}}_i \subseteq \emph{\textbf{h}}_i \in \emph{\textbf{H}}_i$ such that $h \lessdot_i \emph{\textbf{d}}_i$. 
\end{remark}

\medskip

With these notions at hand, we can define two invariant transformations, which are the natural generalizations of the ones provided in Examples \ref{ex:coalescingexample} and \ref{ex:splittingxample}, respectively.


\subsection{Coalescing Moves}\label{subsec:gamma} We represent the so-called \textbf{Coalescing Moves} transformation (or, simply, \textbf{Coalescing}) by the correspondence
$$
\gamma: \mathcal{G}  \twoheadrightarrow \mathcal{G},
$$
and recall that its inverse correspondence is commonly known in the literature as \textbf{Sequential Agent Splitting}.
The domain $\mathrm{dom}(\gamma):=\{G \in \mathcal{G}: \gamma(G) \neq \emptyset\}$ is the collection of all extensive game structures $G$ defined by \eqref{eq:extensive} for which there exist $i \in I$ and information sets $\emph{\textbf{h}}_i,\emph{\textbf{h}}_i^\prime \in \emph{\textbf{H}}_i$ such that 
$\emph{\textbf{h}}_i^\prime$ controls $\emph{\textbf{h}}_i$, that is, 
$\emph{\textbf{h}}_i \ll_i \emph{\textbf{h}}_i^\prime$. 

For each $G \in \mathrm{dom}(\gamma)$ and $\emph{\textbf{h}}_i,\emph{\textbf{h}}_i^\prime \in \emph{\textbf{H}}_i$ with $\emph{\textbf{h}}_i \ll_i \emph{\textbf{h}}_i^\prime$, we identity the game in $\gamma(G)$ corresponding to the pair $(\emph{\textbf{h}}_i,\emph{\textbf{h}}_i^\prime)$ with 
$$
\gamma(G; \emph{\textbf{h}}_i, \emph{\textbf{h}}_i^\prime).
$$ 

Denoting by $a_i^\star \in F_i(\emph{\textbf{h}}_i)$ the (unique) feasible action of the player $i \in I$ for which \eqref{eq:controlling} holds, the transformed game
$$
\gamma(G; \emph{\textbf{h}}_i, \emph{\textbf{h}}_i^\prime):=\langle \tilde{I}, \tilde{\bar{H}}, (\tilde{A}_i, \tilde{\emph{\textbf{H}}}_i)_{i \in \tilde{I}}\rangle
$$
is defined as follows:
\begin{list}{$\bullet$}{}
\item \label{item:c1} $\tilde{I}=I$;
\item \label{item:c2} $\tilde{\bar{H}}$ coincides with $\bar{H}$ for all histories $h$ such that at least one of the following is satisfied:
\begin{enumerate}[label=(\roman*)]
\item there exists $h^\prime \in \emph{\textbf{h}}_i$ such that $h\preceq h^\prime$;
\item $h$ is unrelated to any $h^\prime \in \emph{\textbf{h}}_i$;
\item there exist $h^\prime \in \emph{\textbf{h}}_i$, $a_i \in F_i(h)\setminus \{a_i^\star\}$, and $a_{-i} \in F_{-i}(h)$ with $(h^\prime,(a_i,a_{-i}))\preceq h$.
\end{enumerate}
In the remaining cases, each history $h^\prime \in \bar{H}$ such that $(h,(a_i^\star,a_{-i})) \preceq h^\prime$, for some $a_{-i} \in F_{-i}(h)$, has to be replaced in $\tilde{\bar{H}}$ by $\tilde{h^\prime}$ where the actions chosen by player $i$ at the histories in the information set $\emph{\textbf{h}}_i^\prime$ are shifted back replacing $a_i^\star$; 
\item \label{item:c2} $\tilde{A}_j=A_j$ for all $j \in I$;
\item \label{item:c4} denoting with $\tilde{h}$ the history in $\gamma(G; \emph{\textbf{h}}_i, \emph{\textbf{h}}_i^\prime)$ corresponding to $h \in H$, we have\footnote{We recall by \eqref{eq:disjointinfosets} that $F_i(\emph{\textbf{h}}_i) \cap F_i(\emph{\textbf{h}}_i^\prime)=\emptyset$.}
$$
\tilde{F}_i(\tilde{h})=F_i(\emph{\textbf{h}}_i) \cup F_i(\emph{\textbf{h}}_i^\prime) \setminus \{a_i^\star\}
$$
for all $h \in \emph{\textbf{h}}_i$; otherwise $\tilde{F}_j(\tilde{h})=F_j(h)$. 
The new information sets $(\tilde{\emph{\textbf{H}}}_j)_{j \in I}$ are modified accordingly. In particular, $\emph{\textbf{h}}_i$ and $\emph{\textbf{h}}_i^\prime$ are coalesced into $\tilde{\emph{\textbf{h}}}_i$. Moreover, new information sets of players $j \in I\setminus \{i\}$ who are active in the sub-trees with roots $h \in \emph{\textbf{h}}_i$ are added, so that they are not able to observe the choices made by player $i$ at histories $\tilde{h} \in \tilde{\emph{\textbf{h}}}_i$ among the feasible actions $F_i(\emph{\textbf{h}}_i^\prime)$; see, e.g., the information sets of player \textsl{3} in Figure \ref{fig:coalescinggeneral} below.
\end{list}

\begin{figure}[!htb]
\centering
\begin{minipage}[c]{0.45\textwidth}
\centering
\begin{tikzpicture}
[scale=.95,description/.style=auto]
\node (a0) at (0,8){{\textsl{1}}};
\node (a11) at (-1.5,6){{\textsl{2}}};
\node (a12) at (1.5,6){{\textsl{2}}};
\node (a21) at (-1.5-.9,4){{\footnotesize $z_1$}};
\node (a22) at (-1.5+.9,4){{\textsl{3}}};
\node (a23) at (1.5-.9,4){{\footnotesize $z_2$}};
\node (a24) at (1.5+.9,4){{\textsl{3}}};
\node (a31) at (-.4-1,2){{\textsl{2}}};
\node (a32) at (-.4+.7,2){{\textsl{2}}};
\node (a33) at (2.4-.7,2){{\textsl{2}}};
\node (a34) at (2.4+1,2){{\textsl{2}}};
\node (a411) at (-1.4-.5,0){{\footnotesize $z_3$}};
\node (a412) at (-1.4+.3,0){{\footnotesize $z_4$}};
\node (a421) at (.3-.5,0){{\footnotesize $z_5$}};
\node (a422) at (.3+.3,0){{\footnotesize $z_6$}};
\node (a431) at (1.7-.3,0){{\footnotesize $z_7$}};
\node (a432) at (1.7+.5,0){{\footnotesize $z_8$}};
\node (a441) at (3.4-.3,0){{\footnotesize $z_9$}};
\node (a442) at (3.4+.5,0){{\footnotesize $z_{10}$}};
\draw [dashed, thin] (a11)--(a12);
\draw [densely dotted] (a31)--(a32);
\draw [densely dotted] (a32)--(a33);
\draw [densely dotted] (a33)--(a34);
\draw [-latex,thin] (a0) to node[left]{{\small $\ell$}}(a11);
\draw [-latex,thin] (a0) to node[right]{{\small $r$}}(a12);
\draw [-latex,thin] (a11) to node[left]{{\small $a$}}(a21);
\draw [-latex,thin] (a11) to node[right]{{\small $b$}}(a22);
\draw [-latex,thin] (a12) to node[left]{{\small $a$}}(a23);
\draw [-latex,thin] (a12) to node[right]{{\small $b$}}(a24);
\draw [-latex,thin] (a22) to node[left]{{\small $x$}}(a31);
\draw [-latex,thin] (a22) to node[right]{{\small $y$}}(a32);
\draw [-latex,thin] (a24) to node[left]{{\small $x$}}(a33);
\draw [-latex,thin] (a24) to node[right]{{\small $y$}}(a34);
\draw [-latex,thin] (a31) to node[left]{{\small $p$}}(a411);
\draw [-latex,thin] (a31) to node[right]{{\small $q$}}(a412);
\draw [-latex,thin] (a32) to node[left]{{\small $p$}}(a421);
\draw [-latex,thin] (a32) to node[right]{{\small $q$}}(a422);
\draw [-latex,thin] (a33) to node[left]{{\small $p$}}(a431);
\draw [-latex,thin] (a33) to node[right]{{\small $q$}}(a432);
\draw [-latex,thin] (a34) to node[left]{{\small $p$}}(a441);
\draw [-latex,thin] (a34) to node[right]{{\small $q$}}(a442);
\end{tikzpicture}
\end{minipage}
$\text{ }\hspace{-5mm}\overset{\gamma}{\implies}$
\begin{minipage}[c]{0.45\textwidth}
\centering
\begin{tikzpicture}
[scale=.95,description/.style=auto]
\node (a0) at (0,8){{\textsl{1}}};
\node (a11) at (-1.5,6){{\textsl{2}}};
\node (a12) at (1.5,6){{\textsl{2}}};
\node (a21) at (-1.5-1.3,4){{\footnotesize $z_1$}};
\node (a22a) at (-1.5+.9-1.4,3){{\textsl{3}}};
\node (a22b) at (-1.5+.9+.2,3){{\textsl{3}}};
\node (a23) at (1.5-1.3,4){{\footnotesize $z_2$}};
\node (a24a) at (1.5+.9-1.4,3){{\textsl{3}}};
\node (a24b) at (1.5+.9+.2,3){{\textsl{3}}};
\node (a311) at (-1.5+.9-1.4-.4,1){{\footnotesize $z_3$}};
\node (a312) at (-1.5+.9-1.4+.4,1){{\footnotesize $z_5$}};
\node (a321) at (-1.5+.9+.2-.4,1){{\footnotesize $z_4$}};
\node (a322) at (-1.5+.9+.2+.4,1){{\footnotesize $z_6$}};
\node (a331) at (1.5+.9-1.4-.4,1){{\footnotesize $z_7$}};
\node (a332) at (1.5+.9-1.4+.4,1){{\footnotesize $z_9$}};
\node (a341) at (1.5+.9+.2-.4,1){{\footnotesize $z_8$}};
\node (a342) at (1.5+.9+.2+.4,1){{\footnotesize $z_{10}$}};
\draw [dashed, thin] (a11.7)--(a12.173);
\draw [densely dotted] (a11.350)--(a12.190);
\draw [dashdotted] (a22a)--(a22b);
\draw [dashdotted] (a24a)--(a24b);
\draw [-latex,thin] (a0) to node[left]{{\small $\ell$}}(a11);
\draw [-latex,thin] (a0) to node[right]{{\small $r$}}(a12);
\draw [-latex,thin] (a11) to node[left]{{\small $a$}}(a21);
\draw [-latex,thin] (a11)--(a22a);
\draw [-latex,thin] (a11) to node[right]{{\small $q$}}(a22b);
\draw [-latex,thin] (a12) to node[left]{{\small $a$}}(a23);
\draw [-latex,thin] (a12)--(a24a);
\node (label) at (1.55-.1,4){{\small $p$}};
\node (label) at (-1.45-.1,4){{\small $p$}};
\draw [-latex,thin] (a12) to node[right]{{\small $q$}}(a24b);
\draw [-latex,thin] (a22a) to node[left]{{\small $x$}}(a311);
\draw [-latex,thin] (a22a) to node[right]{{\small $y$}}(a312);
\draw [-latex,thin] (a22b) to node[left]{{\small $x$}}(a321);
\draw [-latex,thin] (a22b) to node[right]{{\small $y$}}(a322);
\draw [-latex,thin] (a24a) to node[left]{{\small $x$}}(a331);
\draw [-latex,thin] (a24a) to node[right]{{\small $y$}}(a332);
\draw [-latex,thin] (a24b) to node[left]{{\small $x$}}(a341);
\draw [-latex,thin] (a24b) to node[right]{{\small $y$}}(a342);
\end{tikzpicture}
\end{minipage}

\vspace{2mm}

\caption{An example of the transformation ``Coalescing'' $\gamma$.}
\label{fig:coalescinggeneral}
\end{figure}

Intuitively, the Coalescing transformation shifts all the actions in a information set $\emph{\textbf{h}}_i^\prime$ of a given player $i$ backwards to another information set $\emph{\textbf{h}}_i$ of $i$ controlled by the first one. 
Note that the histories in $\tilde{\emph{\textbf{h}}^\prime}_i$ corresponding to $\emph{\textbf{h}}^\prime_i$ may ``disappear'' if player $i$ was the only active player in such histories.
Note also that the equalities between action sets that appear in the definition of Coalescing are best interpreted as isomorphisms, since the meaning of actions such as $p$ and $q$ in Figure
\ref{fig:coalescinggeneral} is changed by the transformation. For example, choosing action $p$ on the right-hand side corresponds to choosing $b$ and $p$ on the left-hand side. 

\begin{lemma}\label{lem:gamma}
Coalescing is an invariant transformation, that is, $\gamma \in \mathcal{T}$.
\end{lemma}
\begin{proof}
Fix $G \in \mathrm{dom}(\gamma)$ and consider the transformed game $\gamma(G; \emph{\textbf{h}}_i, \emph{\textbf{h}}_i^\prime)$. 
For each player $j \in I\setminus \{i\}$, there is an obvious isomorphism between $\tilde{\mathcal{S}}_j$ and $\mathcal{S}_j$. 
Moreover, 
there exist a bijection $f_i: \mathcal{S}_i\to \tilde{\mathcal{S}}_i$ and a bijection $g: Z\to \tilde{Z}$ such that $g(\zeta(s))=\tilde{\zeta}(f(s))$ for all $s \in \mathcal{S}$, where $f(s)$ is the structurally reduced strategy profile $(f_j(s_j))_{j \in J}$ and $f_j(s_j):=s_j$ for all $j \in I\setminus \{i\}$. 
This implies that $\mathrm{rn}_Z(G) \simeq \mathrm{rn}_Z(\gamma(G; \emph{\textbf{h}}_i, \emph{\textbf{h}}_i^\prime))$, i.e., $G$ and $\gamma(G; \emph{\textbf{h}}_i, \emph{\textbf{h}}_i^\prime)$ are behaviorally equivalent.
\end{proof}


\subsection{Interchanging/Simultanizing}\label{subsec:iota} 
We denote the \textbf{Interchanging/Simultanizing} transformation by the correspondence
$$
\sigma: \mathcal{G} \twoheadrightarrow \mathcal{G}.
$$
The domain $\mathrm{dom}(\sigma):=\{G \in \mathcal{G}: \sigma(G) \neq \emptyset\}$ is the collection of all extensive game structures $G$ defined as in \eqref{eq:extensive} for which there exist a history $h \in H$, a player $i \in I$, and an information set $\emph{\textbf{h}}_i \in \emph{\textbf{H}}_i$ with a nonempty subset $\emph{\textbf{d}}_i$ such that 
$h \lessdot_i \emph{\textbf{d}}_i$, that is, $\emph{\textbf{d}}_i$ dominates $h$.

For each $G \in \mathrm{dom}(\sigma)$, $h \in H$ and $\emph{\textbf{d}}_i\subseteq \emph{\textbf{h}}_i \in \emph{\textbf{H}}_i$ with $h \lessdot_i \emph{\textbf{d}}_i$, we denote the game in $\sigma(G)$ corresponding to the pair $(h, \emph{\textbf{d}}_i)$ with $
\sigma(G; h, \emph{\textbf{d}}_i)$, where the transformed game
$$
\sigma(G; h, \emph{\textbf{d}}_i):=\langle \tilde{I}, \tilde{\bar{H}}, (\tilde{A}_i,  \tilde{\emph{\textbf{H}}}_i)_{i \in \tilde{I}}\rangle
$$
is defined as follows:
\begin{list}{$\bullet$}{}
\item \label{item:i1} $\tilde{I}=I$;
\item \label{item:i3} $\tilde{\bar{H}}$ coincides with $\bar{H}$ for all histories $h^\prime$ such that either $h^\prime \preceq h$ or $h^\prime$ is unrelated to $h$. In the remaining cases, each history $h^\prime \in \bar{H}$ such that $h \prec h^\prime$ has to be replaced in $\tilde{\bar{H}}$ by $\tilde{h^\prime}$ where the actions chosen by player $i$ at the histories in $\emph{\textbf{d}}_i$ are shifted back at the last coordinate of $h$ (this is possible since, as it been already observed after Definition \ref{def:dominatingset}, player $i$ has to be inactive at $h$); 
\item \label{item:i2} $\tilde{A}_j=A_j$ for all $j \in I$;
\item \label{item:i4} denoting with $\tilde{h^\prime}$ the history in $\sigma(G; h, \emph{\textbf{d}}_i)$ corresponding to $h^\prime \in H$, we have
$$
\tilde{F}_i(\tilde{h})=F_i(\emph{\textbf{d}}_i),
$$
and $\tilde{F}_j(\tilde{h^\prime})=F_j(h^\prime)$ in all remaining cases. 
The new information sets $(\tilde{\emph{\textbf{H}}}_j)_{j \in I}$ are modified accordingly. The position of the subset $\emph{\textbf{d}}_i \subseteq \emph{\textbf{h}}_i$ is shifted back at the last coordinate of $h$, all the others do not change; see e.g. the information set of player \textsl{3} in Figure \ref{fig:interchanginggeneral} below.
\end{list}

Intuitively, the Interchanging/Simultanizing transformation shifts all the actions in a subset $\emph{\textbf{d}}_i$ of an information set $\emph{\textbf{h}}_i$ of a player $i$ backwards to another history $h$ dominated by $\emph{\textbf{d}}_i$ (where he is not active). In addition, as it is shown in Figure \ref{fig:interchanginggeneral}, $\emph{\textbf{d}}_i$ can be a proper subset of $\emph{\textbf{h}}_i$.

\medskip

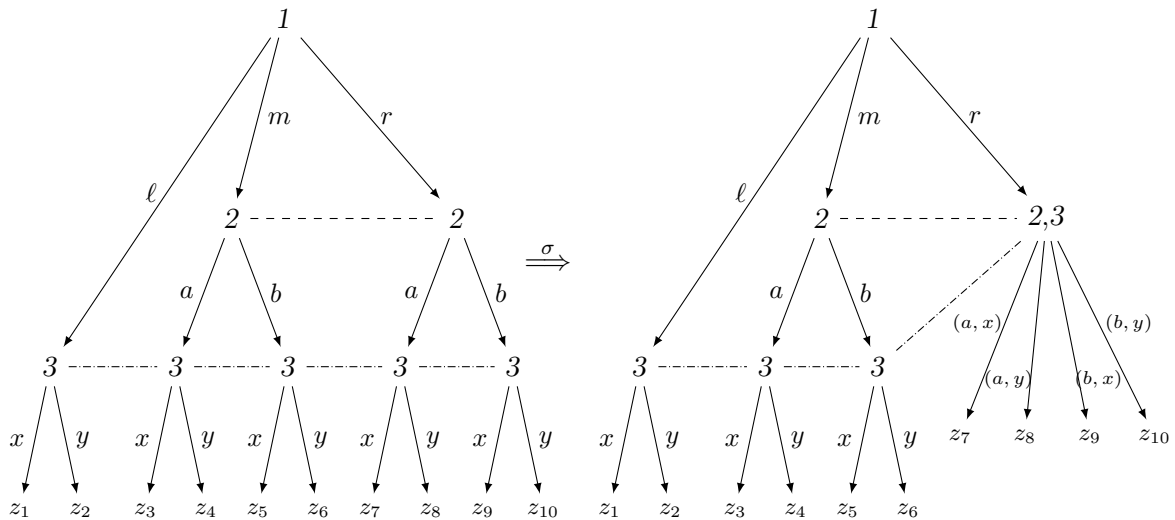
\begin{figure}[!htbp]
\centering
\begin{minipage}[c]{0.45\textwidth}
\centering
\begin{tikzpicture}
[scale=1.15,description/.style=auto]
\node (a0) at (0,8){{\textsl{1}}};
\node (a11) at (-.6,5.7){\textsl{2}};
\node (a12) at (2,5.7){\textsl{2}};
\node (a30) at (-1.5-1.2,4){\textsl{3}};
\draw [dashed] (a11)--(a12);
\draw [-latex,thin] (a0)  to node[right]{{\small $m$}}(a11);
\draw [-latex,thin] (a0)  to node[right]{{\small $r$}}(a12);
\draw [-latex,thin] (a0)  to node[left]{{\small $\ell$}}(a30);
\node (a31) at (-.6-.65,4){\textsl{3}};
\node (a32) at (-.6+.65,4){\textsl{3}};
\node (a33) at (2-.65,4){\textsl{3}};
\node (a34) at (2+.65,4){\textsl{3}};
\draw [densely dashdotted] (a30)--(a31);
\draw [densely dashdotted] (a31)--(a32);
\draw [densely dashdotted] (a32)--(a33);
\draw [densely dashdotted] (a33)--(a34);
\draw [-latex,thin] (a11) to node[left]{{\small $a$}}(a31);
\draw [-latex,thin] (a11) to node[right]{{\small $b$}}(a32);
\draw [-latex,thin] (a12) to node[left]{{\small $a$}}(a33);
\draw [-latex,thin] (a12) to node[right]{{\small $b$}}(a34);
\node (a401) at (-1.5-1.2-.35,2.35){{\footnotesize $z_1$}};
\node (a402) at (-1.5-1.2+.35,2.35){{\footnotesize $z_2$}};
\node (a411) at (-.6-.65-.35,2.35){{\footnotesize $z_3$}};
\node (a412) at (-.6-.65+.35,2.35){{\footnotesize $z_4$}};
\node (a421) at (-.6+.65-.35,2.35){{\footnotesize $z_5$}};
\node (a422) at (-.6+.65+.35,2.35){{\footnotesize $z_6$}};
\node (a431) at (2-.65-.35,2.35){{\footnotesize $z_7$}};
\node (a432) at (2-.65+.35,2.35){{\footnotesize $z_8$}};
\node (a441) at (2+.65-.35,2.35){{\footnotesize $z_9$}};
\node (a442) at (2+.65+.35,2.35){{\footnotesize $z_{10}$}};
\draw [-latex,thin] (a30) to node[left]{{\small $x$}}(a401);
\draw [-latex,thin] (a30) to node[right]{{\small $y$}}(a402);
\draw [-latex,thin] (a31) to node[left]{{\small $x$}}(a411);
\draw [-latex,thin] (a31) to node[right]{{\small $y$}}(a412);
\draw [-latex,thin] (a32) to node[left]{{\small $x$}}(a421);
\draw [-latex,thin] (a32) to node[right]{{\small $y$}}(a422);
\draw [-latex,thin] (a33) to node[left]{{\small $x$}}(a431);
\draw [-latex,thin] (a33) to node[right]{{\small $y$}}(a432);
\draw [-latex,thin] (a34) to node[left]{{\small $x$}}(a441);
\draw [-latex,thin] (a34) to node[right]{{\small $y$}}(a442);
\end{tikzpicture}
\end{minipage}
$\text{ }\hspace{-5mm}\overset{\sigma}{\implies}$
\begin{minipage}[c]{0.45\textwidth}
\centering
\begin{tikzpicture}
[scale=1.15,description/.style=auto]
\node (a0) at (0,8){{\textsl{1}}};
\node (a11) at (-.6,5.7){\textsl{2}};
\node (a12) at (2,5.7){\textsl{2},\textsl{3}};
\node (a30) at (-1.5-1.2,4){\textsl{3}};
\draw [dashed] (a11)--(a12);
\draw [-latex,thin] (a0) to node[right]{{\small $m$}} (a11);
\draw [-latex,thin] (a0) to node[right]{{\small $r$}}(a12);
\draw [-latex,thin] (a0) to node[left]{{\small $\ell$}}(a30);
\node (a31) at (-.6-.65,4){\textsl{3}};
\node (a32) at (-.6+.65,4){\textsl{3}};
\draw [densely dashdotted] (a30)--(a31);
\draw [densely dashdotted] (a31)--(a32);
\draw [densely dashdotted] (a32)--(a12);
\draw [-latex,thin] (a11) to node[left]{{\small $a$}}(a31);
\draw [-latex,thin] (a11) to node[right]{{\small $b$}}(a32);
\node (a401) at (-1.5-1.2-.35,2.35){{\footnotesize $z_1$}};
\node (a402) at (-1.5-1.2+.35,2.35){{\footnotesize $z_2$}};
\node (a411) at (-.6-.65-.35,2.35){{\footnotesize $z_3$}};
\node (a412) at (-.6-.65+.35,2.35){{\footnotesize $z_4$}};
\node (a421) at (-.6+.65-.35,2.35){{\footnotesize $z_5$}};
\node (a422) at (-.6+.65+.35,2.35){{\footnotesize $z_6$}};
\node (a431) at (1,3.2){{\footnotesize $z_7$}};
\node (a432) at (1.35+.4,3.2){{\footnotesize $z_8$}};
\node (a441) at (2.5,3.2){{\footnotesize $z_9$}};
\node (a442) at (2.5+.75,3.2){{\footnotesize $z_{10}$}};
\draw [-latex,thin] (a30) to node[left]{{\small $x$}}(a401);
\draw [-latex,thin] (a30) to node[right]{{\small $y$}}(a402);
\draw [-latex,thin] (a31) to node[left]{{\small $x$}}(a411);
\draw [-latex,thin] (a31) to node[right]{{\small $y$}}(a412);
\draw [-latex,thin] (a32) to node[left]{{\small $x$}}(a421);
\draw [-latex,thin] (a32) to node[right]{{\small $y$}}(a422);
\draw [-latex,thin] (a12)--(a431);
\draw [-latex,thin] (a12)--(a432);
\draw [-latex,thin] (a12)--(a441);
\draw [-latex,thin] (a12)--(a442);
\node (label1) at (1.2, 4.5){{\tiny $(a,x)$}};
\node (label2) at (1.58-.02,4-.15){{\tiny $(a,y)$}};
\node (label3) at (2.56+.04,4-.15){{\tiny $(b,x)$}};
\node (label4) at (2.95,4.5){{\tiny $(b,y)$}};
\end{tikzpicture}
\end{minipage}

\vspace{2mm}

\caption{An example of the transformation ``Interchanging/Simultanizing'' $\sigma$.\label{fig:interchanginggeneral}}
\end{figure}


The extensive game structure in Figure \ref{fig:iotainverse} provides another illustrative example: shifting back the two third-tier nodes of player \textsl{2} we get the extensive game structure in the right hand side of Figure \ref{fig:interchanginggeneral}.

\medskip

\begin{figure}[!htbp]
\centering
\begin{tikzpicture}
[scale=1,description/.style=auto]
\node (a0) at (0,8){{\textsl{1}}};
\node (a11) at (-.6,5.7){\textsl{2}};
\node (a12) at (2,5.7){\textsl{3}};
\node (a30) at (-1.5-1.2,4){\textsl{3}};
\draw [-latex,thin] (a0)  to node[right]{{\small $m$}}(a11);
\draw [-latex,thin] (a0)  to node[right]{{\small $r$}}(a12);
\draw [-latex,thin] (a0)  to node[left]{{\small $\ell$}}(a30);
\node (a31) at (-.6-.65,4){\textsl{3}};
\node (a32) at (-.6+.65,4){\textsl{3}};
\node (a33) at (2-.65,4){\textsl{2}};
\node (a34) at (2+.65,4){\textsl{2}};
\draw [dashed] (a11)--(a33);
\draw [dashed] (a33)--(a34);
\draw [densely dashdotted] (a30)--(a31);
\draw [densely dashdotted] (a31)--(a32);
\draw [densely dashdotted] (a32)--(a12);
\draw [-latex,thin] (a11) to node[left]{{\small $a$}}(a31);
\draw [-latex,thin] (a11) to node[right]{{\small $b$}}(a32);
\draw [-latex,thin] (a12) to node[left]{{\small $x$}}(a33);
\draw [-latex,thin] (a12) to node[right]{{\small $y$}}(a34);
\node (a401) at (-1.5-1.2-.35,2.35){{\footnotesize $z_1$}};
\node (a402) at (-1.5-1.2+.35,2.35){{\footnotesize $z_2$}};
\node (a411) at (-.6-.65-.35,2.35){{\footnotesize $z_3$}};
\node (a412) at (-.6-.65+.35,2.35){{\footnotesize $z_4$}};
\node (a421) at (-.6+.65-.35,2.35){{\footnotesize $z_5$}};
\node (a422) at (-.6+.65+.35,2.35){{\footnotesize $z_6$}};
\node (a431) at (2-.65-.35,2.35){{\footnotesize $z_7$}};
\node (a432) at (2-.65+.35,2.35){{\footnotesize $z_9$}};
\node (a441) at (2+.65-.35,2.35){{\footnotesize $z_8$}};
\node (a442) at (2+.65+.35,2.35){{\footnotesize $z_{10}$}};
\draw [-latex,thin] (a30) to node[left]{{\small $x$}}(a401);
\draw [-latex,thin] (a30) to node[right]{{\small $y$}}(a402);
\draw [-latex,thin] (a31) to node[left]{{\small $x$}}(a411);
\draw [-latex,thin] (a31) to node[right]{{\small $y$}}(a412);
\draw [-latex,thin] (a32) to node[left]{{\small $x$}}(a421);
\draw [-latex,thin] (a32) to node[right]{{\small $y$}}(a422);
\draw [-latex,thin] (a33) to node[left]{{\small $a$}}(a431);
\draw [-latex,thin] (a33) to node[right]{{\small $b$}}(a432);
\draw [-latex,thin] (a34) to node[left]{{\small $a$}}(a441);
\draw [-latex,thin] (a34) to node[right]{{\small $b$}}(a442);
\end{tikzpicture}

\vspace{2mm}

\caption{A $\sigma$-inverse of the game in the right hand side of Figure \ref{fig:interchanginggeneral}.\label{fig:iotainverse}}
\end{figure}
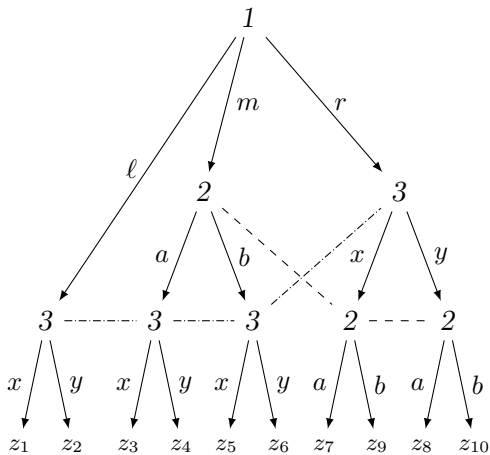

\medskip

Note that the classical \textquotedblleft Interchanging of Simultaneous Moves\textquotedblright\, transformation used in the literature (such as the one represented in Figure \ref{fig:infl}) can be always obtained as follows: $G$ and $G^\prime$ are related by \textquotedblleft Interchanging\textquotedblright\, if 
$G^\prime \in \sigma^{-1}(G^{\prime\prime})$, for some $G^{\prime\prime} \in \sigma(G)$.


\begin{lemma}\label{lem:iota}
Interchanging/Simultanizing is an invariant transformation, that is, $\sigma \in \mathcal{T}$.
\end{lemma}
\begin{proof}
Fix $G \in \mathrm{dom}(\sigma)$ and consider the transformed game $\sigma(G; h, \emph{\textbf{d}}_i)$. It follows by construction that $S_j$ is isomorphic to $\tilde{S}_j$ for all $j \in I$. In particular, $\mathcal{S}_j$ is isomorphic to $\tilde{\mathcal{S}}_j$ for all $j \in I$, so that 
there exist bijections $f_j: \mathcal{S}_j\to \tilde{\mathcal{S}}_j$ and a bijection $g: Z\to \tilde{Z}$ such that $g(\zeta(s))=\tilde{\zeta}(f(s))$ for all $s \in \mathcal{S}$, where $f(s)$ is the structurally reduced strategy profile $(f_j(s_j))_{j \in I}$. 
In particular, $\mathrm{rn}_Z(G) \simeq \mathrm{rn}_Z(\sigma(G; h, \emph{\textbf{d}}_i))$, i.e., $G$ and $\sigma(G; h, \emph{\textbf{d}}_i)$ are behaviorally equivalent.
\end{proof}

\begin{remark}\label{remark:compositioninvarianttransformations}
It is immediate to see that the composition of invariant transformations is invariant.  
In particular, by Lemma \ref{lem:gamma} and Lemma \ref{lem:iota}, the compositions of transformations $\gamma$ and $\sigma$ and their inverses are invariant. 
First, define $ \overset{\mathrm{iso}}{\in}$ as inclusion up to isomorphisms. 
Then, a transformation $T: \mathcal{G} \twoheadrightarrow \mathcal{G}$ with $\mathrm{dom}(T)\neq \emptyset$ is invariant provided that
\begin{equation}\label{eq:condition}
\forall G \in \mathcal{G},\forall G^\prime \in T(G), \exists n \in \mathbb{N}_0, \exists \iota_1,\ldots,\iota_n \in \{\gamma, \gamma^{-1}, \sigma, \sigma^{-1}\},\,\,\,\,\,\,\,\, G^\prime \overset{\mathrm{iso}}{\in} (\iota_1 \circ \cdots \circ \iota_n)(G),
\end{equation}
where, by convention, $(\iota_1 \circ \cdots \circ \iota_n)(G):=\{G\}$ if $n=0$.
\end{remark}


\section{Characterization of Behavioral Equivalence}\label{sec:mainth}

%
Fix a game structure $G \in\mathcal{G}$. Then we say that:
\begin{itemize}
\item $G$ has a \textbf{coalescing opportunity} if we can find in $G$ two
information sets $\emph{\textbf{h}}_{i},\emph{\textbf{h}}_{i}^{\prime }$ of the same
player $i$ such that $\emph{\textbf{h}}_{i}^{\prime }$ controls $\emph{\textbf{h}}_{i}$, that is, $G \in \mathrm{dom}(\gamma)$;
\item $G$ has a \textbf{simultanizing opportunity} if we can find in $G$ a
non terminal history $h$ and a subset $\emph{\textbf{d}}_{i}$ of an information set
of a player $i$ (not active at $h$) such that $\emph{\textbf{d}}_{i}$ dominates $h$, that is, $G \in \mathrm{dom}(\sigma)$;
\item $G$ is \textbf{minimal} if it has no coalescing or simultanizing opportunity, that is, 
$$
G \in \widehat{\mathcal{G}}:=\mathcal{G}\setminus (\mathrm{dom}(\sigma) \cup \mathrm{dom}(\gamma)).
$$
\end{itemize}
In addition,
\begin{itemize}
\item A game structure $G^{\prime} \in \mathcal{G}$ is a \textbf{reduction} of $G$ if there is a finite sequence of game structures $\left(G_1,\ldots,G_n\right)$, with $n \in \mathbb{N}$ such that $G_{1}=G$, $G_n=G^{\prime}$, and $G_{k+1}\in \gamma(G_{k}) \cup \sigma(G_k)$ for each $k=1,\ldots,n-1$; in particular, $G$ is a reduction of $G$ itself.\footnote{We use the term \textquotedblleft reduction\textquotedblright\, because, for both transformations, the height of the tree decreases (in a weak sense). In addition, note that transformation $\gamma$ reduces strictly the sum of the cardinalities of the action sets of the involved player.}
\end{itemize}

Note that, if we interpret $\gamma $ and $\sigma $ as binary relations on $\mathcal{G}$ (with $G\,\gamma\, G^{\prime }$ if $G^{\prime}\in \gamma \left(G\right)$ and $G\,\sigma\, G^{\prime }$ if $G^{\prime }\in \sigma \left( G\right) $) 
and let $\left(\gamma \cup \sigma \right) ^{\star}$ denote the transitive closure of relation $\gamma \cup \sigma$, 
then $G^{\prime }$ is a reduction of $G$ if and only if $G=G^\prime$ or $G\left( \gamma \cup \sigma \right) ^{\star}G^{\prime }$.

With these premises, we are going to show in Lemma \ref{lem:minimalextensive} 
below that each game structure $G \in \mathcal{G}$ has a unique minimal reduction. 
Then our main result follows: 
\begin{theorem}\label{th:mainresult}
For all extensive game structures $G, G \in \mathcal{G}$, the following are equivalent:%
\begin{enumerate}[label={\rm (\textsc{a}\arabic{*})}]
\item \label{item:a1} $G$ and $G^\prime$ are behaviorally equivalent; 
\item \label{item:a2} $G$ and $G^\prime$ have the same minimal reduction; 
\item \label{item:a3} $G$ and $G^\prime$ have a common reduction, up to
isomorphisms;
\item \label{item:a4} $G$ can be transformed into $G^\prime$, up to isomorphisms, through a 
\textup{(}possibly empty\textup{)} 
finite chain of transformations $\gamma$ and $\sigma$ and their inverses.
\end{enumerate}
\end{theorem}


Before we prove our characterization, note that the order on which the two transformations $\gamma$ and $\sigma$ are applied matters. Indeed, as we show in Example \ref{example:independent}, it is possible that $G \in \mathrm{dom}(\gamma)\, \cap\, \mathrm{dom}(\sigma)$ and, at the same time, $\gamma(G) \cap \mathrm{dom}(\sigma)=\emptyset$.
\begin{example}\label{example:independent}
Consider the game structure $G$ in left hand side of Figure \ref{fig:noindependent}. Transformation $\sigma$ may be applied to the subgame with root $(t)$. However, considering that the bottom information set of player \textsl{1}, say $\emph{\textbf{h}}_{\textsl{1}}$, controls $\{(t)\}$, we may apply also the $\gamma$ transformation, which yields the game structure $\gamma(G; \{(t)\}, \emph{\textbf{h}}_{\textsl{1}})$ in the right hand side of Figure \ref{fig:noindependent}. It can be checked that $\sigma$ cannot be applied to $\gamma(G; \{(t)\}, \emph{\textbf{h}}_{\textsl{1}})$, that is, $\gamma(G; \{(t)\}, \emph{\textbf{h}}_{\textsl{1}}) \notin \mathrm{dom}(\sigma)$.

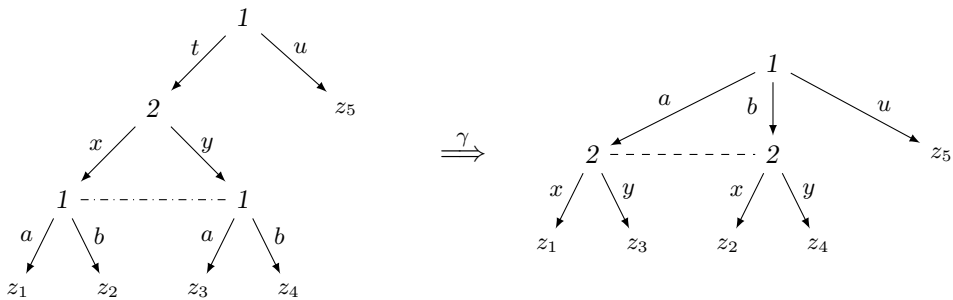
\begin{figure}[!htbp]
\centering
\begin{minipage}[c]{0.41\textwidth}
\centering
\begin{tikzpicture}
[scale=.8,description/.style=auto]
\node (a00) at (1.5,4.5){{\small {\textsl{1}}}};
\node (a01) at (0,3){{\small {\textsl{2}}}};
\node (a11) at (-2+.5,1.5){{\small {\textsl{1}}}};
\node (a12) at (2-.5,1.5){{\small {\textsl{1}}}};
\node (a13) at (3.2,3){{\footnotesize $z_5$}};
\node (a21) at (1.5-.75,0){{\footnotesize $z_3$}};
\node (a22) at (1.5+.75,0){{\footnotesize $z_4$}};
\node (a31) at (-1.5-.75,0){{\footnotesize $z_1$}};
\node (a32) at (-1.5+.75,0){{\footnotesize $z_2$}};
\draw [-latex,thin](a00)--(a01);
\draw [-latex,thin](a00)--(a13);
\draw [-latex,thin](a01)--(a11);
\draw [-latex,thin](a01)--(a12);
\draw [-latex,thin](a12)--(a21);
\draw [-latex,thin](a12)--(a22);
\draw [-latex,thin](a11)--(a31);
\draw [-latex,thin](a11)--(a32);
\draw [dashdotted](a11)--(a12);
\node (label01) at (-1.3+.3+.7-.65,2.4){{\footnotesize $x$}};
\node (label02) at (.9,2.4){{\footnotesize $y$}};
\node (label03) at (.9,.9){{\footnotesize $a$}};
\node (label04) at (2.1,.9){{\footnotesize $b$}};
\node (label05) at (.9-3,.9){{\footnotesize $a$}};
\node (label06) at (2.1-3,.9){{\footnotesize $b$}};
\node (label07) at (1.5-.8,4){{\footnotesize $t$}};
\node (label08) at (1.5+.95,4){{\footnotesize $u$}};
\end{tikzpicture}
\end{minipage}
$\overset{\gamma}{\implies}$
\begin{minipage}[c]{0.41\textwidth}
\centering
\begin{tikzpicture}
[scale=.8,description/.style=auto]
\node (a01) at (1.5,3){{\small {\textsl{1}}}};
\node (a11) at (-2+.5,1.5){{\small {\textsl{2}}}};
\node (a12) at (2-.5,1.5){{\small {\textsl{2}}}};
\node (a13) at (4.3,1.5){{\footnotesize $z_5$}};
\node (a21) at (1.5-.75,0){{\footnotesize $z_2$}};
\node (a22) at (1.5+.75,0){{\footnotesize $z_4$}};
\node (a31) at (-1.5-.75,0){{\footnotesize $z_1$}};
\node (a32) at (-1.5+.75,0){{\footnotesize $z_3$}};
\draw [-latex,thin](a01)--(a11);
\draw [-latex,thin](a01)--(a12);
\draw [-latex,thin](a01)--(a13);
\draw [-latex,thin](a12)--(a21);
\draw [-latex,thin](a12)--(a22);
\draw [-latex,thin](a11)--(a31);
\draw [-latex,thin](a11)--(a32);
\draw [dashed,thin](a11)--(a12);
\node (label01) at (-1.3+.3+.7,2.4){{\footnotesize $a$}};
\node (label02) at (1.15,2.3){{\footnotesize $b$}};
\node (label07) at (3.35,2.3){{\footnotesize $u$}};
\node (label03) at (.9,.9){{\footnotesize $x$}};
\node (label04) at (2.1,.9){{\footnotesize $y$}};
\node (label05) at (.9-3,.9){{\footnotesize $x$}};
\node (label06) at (2.1-3,.9){{\footnotesize $y$}};
\end{tikzpicture}
\end{minipage}
\caption{$\gamma$ and $\sigma$ are not independent.\label{fig:noindependent}}
\end{figure}
\end{example}


While standard reduced normal forms defined in terms of payoffs do not allow the identification of a game tree, for $Z$-reduced normal forms we can instead identify an \textquotedblleft essentially unique\textquotedblright\, (that is, up to isomorphisms) minimal game structure. The next result consists in obtaining such essentially unique game.

\begin{lemma}\label{lem:lemmaisomorphic}
Let $\widehat{G}$ and $\widehat{G}^\prime$ be two behaviorally equivalent minimal game structures. Then $\widehat{G}=\widehat{G}^\prime$, up to isomorphisms.
\end{lemma}
\begin{proof}
Let $\mathcal{G}_0$ be the set of extensive game structures that are behaviorally equivalent to $\widehat{G}$, that is, 
$$
\mathcal{G}_0:=\left\{G \in \mathcal{G}: \mathrm{rn}_Z(G)\simeq \mathrm{rn}_Z(\widehat{G})\right\}.
$$
It is claimed that, up to isomorphisms, $\mathcal{G}_0 \cap \widehat{\mathcal{G}}$ is a singleton.

Fix a \emph{minimal} extensive game structure 
$
G_0
\in \mathcal{G}_0 \cap \widehat{\mathcal{G}}. 
$ 
First, we claim that, up to isomorphisms, at the root of $G_0$ there are the same active players as at the root of $\widehat{G}$ with the same feasible actions. 

By Lemma \ref{partitions}, for each game $G$ and for each player $i$, we can define a partition $\mathscr{P}_{i}^\star$ of $\mathcal{S}_i$ as follows:
\begin{equation}\label{eq:suppartition}
\textstyle \mathscr{P}_{i}^\star:=\sup \left\{\mathscr{P}_i \in \bm{\mathscr{P}}\text{{\scriptsize art}}(\mathcal{S}_i): \{\tilde{\zeta}(P_i\times \mathcal{S}_{-i}): P_i \in \mathscr{P}_i\}\in \bm{\mathscr{P}}\text{{\scriptsize art}}(Z)\right\}.
\end{equation}
(Indeed, the above collection contains $\{\mathcal{S}_i\}$, hence it is nonempty.) In other words, $\mathscr{P}_{i}^\star$ is the finest partition of the set of structurally reduced strategies $\mathcal{S}_i$ of player $i$ which, independently of the profile of strategies of the opponents $s_{-i}\in \mathcal{S}_{-i}$, induces a partition of the terminal histories through the map $\tilde{\zeta}$. 
In particular, we will prove that, if $i$ is active at $\varnothing$ (that is, $i \in I(\varnothing)$), then the partitions of $\mathcal{S}_i$ and $Z$ correspond to his set of feasible actions $F_i(\varnothing)$; whereas if $i \notin I(\varnothing)$ the partitions of $\mathcal{S}_i$ and $Z$ are trivial. 

For each player $i \in I$ and information set $\emph{\textbf{h}}_{i} \in \emph{\textbf{H}}_i$, define the partition 
\begin{equation}\label{eq:partitiondefinition}
\mathscr{P}_i(\emph{\textbf{h}}_{i}):=\{P_i(\emph{\textbf{h}}_{i},a_i): a_i \in F_i(\emph{\textbf{h}}_{i})\}\in \bm{\mathscr{P}}\text{{\scriptsize art}}(\mathcal{S}_i(\emph{\textbf{h}}_{i})),
\end{equation}
where $\mathcal{S}_i(\emph{\textbf{h}}_{i})$ denotes the set of structurally reduced strategies $s_i \in \mathcal{S}_i$ which are consistent with $\emph{\textbf{h}}_{i}$, and 
$P_i(\emph{\textbf{h}}_{i},a_i)$ denotes the set of strategies $s_i \in \mathcal{S}_i(\emph{\textbf{h}}_{i})$ which select $a_i$ at $\emph{\textbf{h}}_{i}$. Note also that $\mathscr{P}_i(\emph{\textbf{h}}_{i})$ induces the following partition:
\begin{equation}\label{eq:partitionZ}
\mathscr{Z}_i(\emph{\textbf{h}}_{i}):=\left\{\tilde{\zeta}(P_i(\emph{\textbf{h}}_{i},a_i)\times \mathcal{S}_{-i}(\emph{\textbf{h}}_{i})): P_i(\emph{\textbf{h}}_{i},a_i) \in \mathscr{P}_i(\emph{\textbf{h}}_{i})\right\} \in \bm{\mathscr{P}}\text{{\scriptsize art}}(Z(\emph{\textbf{h}}_{i})).
\end{equation}
In particular, $|\mathscr{Z}_i(\emph{\textbf{h}}_{i})|=|\mathscr{P}_i(\emph{\textbf{h}}_{i})|=|F_i(\emph{\textbf{h}}_i)|$.

\begin{claim}\label{claim1}
For each game $G$ and for each player $i$, 
if $i \in I(\varnothing)$ then $\mathscr{P}_{i}^\star\neq \{\mathcal{S}_i\}$.
\end{claim}
\begin{proof}
Since $\{\varnothing\} \in \emph{\textbf{H}}_i$, $\mathcal{S}_i(\{\varnothing\})=\mathcal{S}_i$, and $Z(\{\varnothing\})=Z$, we have by \eqref{eq:partitiondefinition} and \eqref{eq:partitionZ} that $\mathscr{P}_i(\{\varnothing\}) \in \bm{\mathscr{P}}\text{{\scriptsize art}}(\mathcal{S}_i)$ and $\mathscr{Z}_i(\{\varnothing\}) \in \bm{\mathscr{P}}\text{{\scriptsize art}}(Z)$. Hence 
$|\mathscr{P}_i^\star| \ge |\mathscr{P}_i(\{\varnothing\})|=|F_i(\varnothing)|\ge 2$.
\end{proof}

\begin{claim}\label{claim2}
For each game $G$ and for each player $i$, 
if $\mathscr{P}_{i}^\star\neq \{\mathcal{S}_i\}$ then 
there exists $\emph{\textbf{h}}_i \in \emph{\textbf{H}}_i$ such that 
$Z(\emph{\textbf{h}}_i)=Z$.
\end{claim}
\begin{proof}
We prove the claim by induction on the lenght $\ell$ of the longest history of the game. 
If the game has longest history $\ell=1$, then $\emph{\textbf{h}}_i=\{\varnothing\}$ and the thesis is obvious. 
Suppose that the claim holds for all games such that the lenght of their longest history is at most $\ell \in \mathbb{N}$. 
Then let $G$ be a game structure such that the lenght of its longest history is $\ell+1$. If $i\in I(\varnothing)$ then, by Claim \ref{claim1}, $\mathscr{P}_{i}(\{\varnothing\})\neq \{\mathcal{S}_i\}$ and $Z(\{\varnothing\})=Z$. 

Suppose now that $i \notin I(\varnothing)$. In this case, let $\{h^{(1)},\ldots,h^{(k)}\}$ be the set of histories of lenght $1$ of $G$, so that $h^{(j)}=(a)$, with $a \in \prod_{\iota \in I(\varnothing)}F_\iota(\varnothing)$; also, for each $j=1,\ldots,k$, let $G^{(j)}$ be the \textquotedblleft subgame\textquotedblright\, which starts at $h^{(j)}$ defined by restriction on the subtree with root $h^{(j)}$ (cutting the information sets, if necessary). 
For each $j=1,\ldots,k$, the lenght of the longest history in $G^{(j)}$ is at most $\ell$, and by the induction hypothesis there exists a unique information set $\emph{\textbf{h}}_i^{(j)}$ of $i$ in $G^{(j)}$ for which $Z(\emph{\textbf{h}}_i^{(j)})=Z(h^{(j)})$. Let $\{P_1,\ldots,P_m\} \in \bm{\mathscr{P}}\text{{\scriptsize art}}(\mathcal{S}_i)$ be a partition of $\mathcal{S}_i$ into nonempty sets such that 
$$
\{\tilde{\zeta}(P_1\times \mathcal{S}_{-i}), \ldots, \tilde{\zeta}(P_m\times \mathcal{S}_{-i})\} \in \bm{\mathscr{P}}\text{{\scriptsize art}}(Z).
$$
Also, for each $j=1,\ldots,k$, let $\mathcal{S}_{-i}^{(j)}$ be the set of profiles of structurally reduced strategies $s_{-i} \in \mathcal{S}_{-i}$ consistent with history $h^{(j)}$. Then $\{\tilde{\zeta}(P_1\times \mathcal{S}_{-i}^{(j)}), \ldots, \tilde{\zeta}(P_m\times \mathcal{S}_{-i}^{(j)})\} \in \bm{\mathscr{P}}\text{{\scriptsize art}}(Z(h^{(j)}))$. 
It follows that the union of all information sets $\emph{\textbf{h}}_i^{(j)}$ of the \textquotedblleft subgames\textquotedblright\, $G^{(j)}$ is a unique information set of the original game $G$, so that $Z(\emph{\textbf{h}}_i)=\bigcup_{j=1}^kZ( \emph{\textbf{h}}_i^{(j)})=Z$.
\end{proof}

\begin{claim}\label{claim3}
If $\mathscr{P}_{i}^\star\neq \{\mathcal{S}_i\}$ in the minimal game $G_0$ then $i \in I(\varnothing)$. 
\end{claim}
\begin{proof}
Suppose by way of contradiction that $i \notin I(\varnothing)$ and $\mathscr{P}_{i}^\star\neq \{\mathcal{S}_i\}$. 
By Claim \ref{claim2}, there exists $\emph{\textbf{h}}_i \in \emph{\textbf{H}}_i$ such that 
$Z(\emph{\textbf{h}}_i)=Z$. 
Since $i\notin I(\varnothing)$, $\emph{\textbf{h}}_i \neq \{\varnothing\}$; furthermore, $\emph{\textbf{h}}_i$ dominates $\varnothing$. It follows that there is a simultanizing opportunity at $\varnothing$. This contradicts the minimality of $G_0$.
\end{proof}

It follows by Claims \ref{claim1}, \ref{claim2}, and \ref{claim3} that a player $i \in I$ is active at the root $\varnothing$ of $G_0$ if and only if $\mathscr{P}_{i}^\star\neq \{\mathcal{S}_i\}$. Then, we show that, if player $i$ is active, his set of feasible actions $F_i(\varnothing)$ can be labeled by the elements of $\mathscr{P}_{i}^\star$ (cf. Example \ref{ex:reconstruction} for an illustrative example):
\begin{claim}\label{claim4}
For each $i \in I(\varnothing)$ in the minimal game $G_0$, we have $|F_i(\varnothing)|=|\mathscr{P}_{i}^\star|$.
\end{claim}
\begin{proof}
Fix $i \in I(\varnothing)$. By definition $\mathscr{P}_{i}^\star$ is finer than $\mathscr{P}_i(\{\varnothing\})$ (defined in Equation \eqref{eq:partitiondefinition}), i.e., $\mathscr{P}_i(\{\varnothing\}) \le \mathscr{P}_{i}^\star$ and, in particular, $2\le |F_i(\varnothing)| \le |\mathscr{P}_i^\star|$. 
Suppose by way of contradiction that the second inequality is strict, i.e., $\mathscr{P}_i(\{\varnothing\}) \neq  \mathscr{P}_{i}^\star$. 
Then there exists a partition $\mathscr{P}_i \in \bm{\mathscr{P}}\text{{\scriptsize art}}(\mathcal{S}_i)$ such that $\{\tilde{\zeta}(P_i\times \mathcal{S}_{-i}): P_i \in \mathscr{P}_i\}\in \bm{\mathscr{P}}\text{{\scriptsize art}}(Z)$ and $\mathscr{P}_i \not\le \mathscr{P}_i(\{\varnothing\})$. In particular, $\mathscr{P}_i \neq \{\mathcal{S}_i\}$. 
However, by Lemma \ref{partitions}, $\bm{\mathscr{P}}\text{{\scriptsize art}}(\mathcal{S}_i)$ is a complete lattice and there exists 
$$
\mathscr{P}_i^\prime:=\sup\{\mathscr{P}_i(\{\varnothing\}), \mathscr{P}_i\} \in \bm{\mathscr{P}}\text{{\scriptsize art}}(\mathcal{S}_i).
$$
As in the proof of Claim \ref{claim2}, let $\{h^{(1)},\ldots,h^{(k)}\}$ be the set of histories of lenght $1$ of $G_0$, so that $h^{(j)}=(a)$, with $a \in \prod_{\iota \in I(\varnothing)}F_\iota(\varnothing)$; similarly, for each $j=1,\ldots,k$, let $G_0^{(j)}$ be the \textquotedblleft subgame\textquotedblright\, which starts at $h^{(j)}$ defined by restriction on the subtree with root $h^{(j)}$ (cutting the information sets, if necessary). 
Note that each $G_0^{(j)}$ must be minimal like $G_0$ (indeed, if there were a simultanizing or a coalescing opportunity in $G_0^{(j)}$, there would be also in $G_0$). 
Considering that $\mathscr{P}_i(\{\varnothing\}) \le \mathscr{P}_i^\prime$, $\mathscr{P}_i(\{\varnothing\}) \neq \mathscr{P}_i^\prime$, and that $\{\tilde{\zeta}(P_i^\prime\times \mathcal{S}_{-i}): P_i^\prime \in \mathscr{P}_i^\prime\}\in \bm{\mathscr{P}}\text{{\scriptsize art}}(Z)$, it follows that there exists $j \in \{1,\ldots,k\}$ such that the partition of the corresponding \textquotedblleft subgame\textquotedblright\, $G_0^{(j)}$ is not a singleton. 
Hence, by Claim \ref{claim3}, player $i$ is active at the root of $G_0^{(j)}$. 
Thus $\{\varnothing\}\ll_i \emph{\textbf{h}}_i$, where $\emph{\textbf{h}}_i$ is the information set of $i$ containing $h^{(j)}$. Therefore there is a coaleascing opportunity at $\varnothing$. This contradicts the minimality of $G_0$. 
\end{proof}

\bigskip


It follows that $G_0$ (hence, also $\widehat{G}^\prime$) has, up to isomorphisms, the same active players at the root and the same feasible actions of $\widehat{G}$.



\medskip

Then, for each $i \in I$, fix $P_i^\star \in \mathscr{P}_{i}^\star$. 
For each $i \in I(\varnothing)$, $P_i^\star$ identifies an action which is feasible at $\varnothing$. For each $i \notin I(\varnothing)$, $\mathscr{P}_{i}^\star$ is the trivial partition by Claim \ref{claim1}, hence $P_i^\star=\mathcal{S}_i$ indicates that no action is available at $\varnothing$. 
Notice that, with this, $(P_i^\star)_{i \in I}$ identifies an history $h$ of lenght $1$. 
Consider the sub-$Z$-reduced normal form, where the set of strategies is restricted to $P_i^\star$ for each $i \in I$. Let $Z^\star$ be the set of terminal histories consistent with $\prod_{j\in I} P_j^\star$. With the same argument as above, for each $i \in I$, there exists an unique partition $\mathscr{P}_{i}^{\star\star}$ of $P_i^\star$ defined by
$$
\textstyle \mathscr{P}_{i}^{\star\star}:=\sup \left\{\mathscr{P}_i \in \bm{\mathscr{P}}\text{{\scriptsize art}}(P_i^\star): \{\tilde{\zeta}\left(P_i\times \prod_{j \in I\setminus \{i\}}P_j^\star\right): P_i \in \mathscr{P}_i\}\in \bm{\mathscr{P}}\text{{\scriptsize art}}(Z^\star)\right\}.
$$
Similarly, player $j \in I$ is active at the history $h$ if and only if $\mathscr{P}_j^{\star\star} \neq \{P_j^\star\}$ and his set of feasible actions $F_j(h)$ can be labeled as the elements of $\mathscr{P}_j^{\star\star}$. 
At this point, for each player $i$, if two histories $h^\prime, h^{\prime\prime}$ of lenght $\le 1$ belong to the same information set of $i$, then they identify the same partition $\mathscr{P}_{i}^{\star\star}$ (hence, by construction, the same feasible actions). Conversely, if the histories $h^\prime, h^{\prime\prime}$ identify the same partition $\mathscr{P}_{i}^{\star\star}$ then by Claim \ref{claim4} they belong to the same information set.

\medskip

This algorithm can be inductively repeated, 
relabeling all profiles of actions of $G_0$ by profiles of cells of partitions of each $\mathcal{S}_i$. By construction, with this relabeling, two histories $h,h^\prime$ are on a same information set of a player $i$ if and only if 
$F_i(h)=F_i(h^\prime)$.
%
%
This allows to reconstruct uniquely (that is, up to isomorphisms) the whole game structure $G_0$, completing the proof. 
\end{proof}

\begin{lemma}\label{lem:minimalextensive}
Each extensive game structure $G \in \mathcal{G}$ has a unique minimal reduction, up to isomorphisms. 
\end{lemma}
\begin{proof} 
Fix $G \in \mathcal{G}$. First, we show that there exists a minimal game structure $\widehat{G}$ which is behaviorally equivalent to $G$. 
The game $\widehat{G}$ is constructed by recursive application of transformations $\gamma$ and $\sigma$. The recursion is based on the length of histories that offer coalescing or simultanizing opportunities. Recall that we consider only finite game structures; therefore, all the collections defined below are finite as well.

\medskip

\textsc{Basis step.} Consider the root $\varnothing$, that is, the only history of length $0$. 
Define the set $\mathscr{C}_0$ made by all pairs $(\{\varnothing\}, \emph{\textbf{h}}_i)$ such that $\{\varnothing\}\ll_i \emph{\textbf{h}}_i$ (so that $i \in I(\varnothing)$ and $\{\varnothing\} \in \emph{\textbf{H}}_i$). 
Define also the set $\mathscr{S}_0$ made by all pairs $(\varnothing, \emph{\textbf{d}}_i)$ such that $\varnothing\lessdot_i \emph{\textbf{d}}_i$ (so that $i\notin I(\varnothing)$). 
Of course, both sets $\mathscr{C}_0$ and $\mathscr{S}_0$ are finite. 
Apply the transformation $\gamma$ at all pairs $(\{\varnothing\}, \emph{\textbf{h}}_i) \in \mathscr{C}_0$; 
to be precise, enumerate the elements of $\mathscr{C}_0$ as $(\{\varnothing\},w_1),\ldots,(\{\varnothing\},w_k)$, in some order, and consider the finite sequence of games $\gamma(G; \{\varnothing\},w_1)$, $\gamma(\gamma(G; \{\varnothing\},w_1); \{\varnothing\},w_2^\prime)$  (where $(\{\varnothing\},w_2^\prime)$ is the pair corresponding to $(\{\varnothing\},w_2)$ in $\gamma(G; \{\varnothing\},w_1)$), $\ldots$, 
$\gamma(\cdots \gamma(\gamma(G; \{\varnothing\},w_1); \{\varnothing\},w_2^\prime)\cdots; \{\varnothing\},w_k^\prime)$ and note that the latter game does not depend on the chosen enumeration of $\mathscr{C}_0$. 
Then, similarly, apply the transformation $\sigma$ in the latter game at the corresponding pairs $(\varnothing, \emph{\textbf{d}}_i) \in \mathscr{S}_0$ and denote the obtained game by $G(1)$ (where $G(1)=G$ if both $\mathscr{C}_0$ and $\mathscr{S}_0$ are empty). Note that, since all transformations are applied at the root of the game, all coalescing and simultanizing opportunities are preserved at each transformation.

\medskip

\textsc{Induction Step.} Suppose that $G(n)$ has been defined for some positive integer $n$. 
Define the set $\mathscr{C}_n$ made by all pairs $(\emph{\textbf{h}}_i, \emph{\textbf{h}}_i^\prime)$ of the game $G(n)$ such that $\emph{\textbf{h}}_i\ll_i \emph{\textbf{h}}_i^\prime$ and there exists $h \in \emph{\textbf{h}}_i$ with lenght $n$ (in particular, $i \in I(h)$). 
Define also the set $\mathscr{S}_n$ made by all pairs $(h, \emph{\textbf{d}}_i)$ of the game $G(n)$ such that $h\lessdot_i \emph{\textbf{d}}_i$ and $h$ has lenght $n$ (so that $i\notin I(h)$ for each such history $h$). 
Apply the transformation $\gamma$ at all pairs $(\emph{\textbf{h}}_i, \emph{\textbf{h}}_i^\prime) \in \mathscr{C}_n$. Then, apply the transformation $\sigma$ at all pairs $(h, \emph{\textbf{d}}_i)\in \mathscr{S}_n$ and denote by $G(n+1)$ the obtained game (where $G(n+1)=G(n)$ if both $\mathscr{C}_n$ and $\mathscr{S}_n$ are empty). Similarly, note that, since all transformation are applied at the same height of $G(n)$, all coalescing and simultanizing opportunities are preserved at each transformation.

\medskip

By Lemma \ref{lem:gamma} and Lemma \ref{lem:iota}, $G(n)$ is behaviorally equivalent to $G$ for every $n\ge 1$. In addition, since $G$ is finite, there exists an integer $n_0 \ge 1$ such that $G(n)=G(n_0)$ for all $n\ge n_0$. Set $\widehat{G}:=G(n_0)$. We claim that $\widehat{G}\in \widehat{\mathcal{G}}$. 
Let us suppose by contradiction that $\widehat{G} \in \mathrm{dom}(\gamma)$, i.e., there exists a coalescing opportunity in $\widehat{G}$, let us say with $\emph{\textbf{h}}_i \ll_i \emph{\textbf{h}}_i^\prime$. Let $k$ be the minimal lenght of histories $h \in \emph{\textbf{h}}_i$ and note that $k \in \{0,1,\ldots,n_0-1\}$. This implies that $G(k+1)$ has a coalescing oppurtinity at an information set with an history of lenght $k$, which contradicts our costruction. With a similar argument, we have $\widehat{G} \notin \mathrm{dom}(\sigma)$. Therefore $\widehat{G}$ is minimal, hence it is a reduction of $G$.

\medskip

To complete the proof, we need to show that $\widehat{G}$ is the unique reduction of $G$, up to isomorphisms. Indeed, suppose that $\widehat{G}^\prime$ is another minimal game structure which is behaviorally equivalent to $G$. Then $\widehat{G}$ and $\widehat{G}^\prime$ are two behaviorally equivalent minimal game structures. 
It follows from Lemma \ref{lem:lemmaisomorphic} that $\widehat{G}=\widehat{G}^\prime$, up to isomorphisms. 
\end{proof}

\begin{example}\label{ex:reconstruction} Consider the game $\widehat{G} \in \widehat{\mathcal{G}}$ in the right hand side of Figure \ref{fig:noindependent}. Its $Z$-reduced normal form is given in Figure \ref{fig_BoS_dissipative}.


\begin{figure}[!htbp]
\centering
\begin{tikzpicture}
[scale=1.5,description/.style=auto]
\node(d3) at (0,0){
\begin{tabular}{|c|c|c|}
\hline
{\small \textsl{1}\textbackslash \textsl{2}} & {\small $\text{ }$\,\,\,$x$\,\,\,$\text{ }$} & {\small $\text{ }$\,\,\,$y$\,\,\,$\text{ }$}  \\ \hline
{\small $a$} & {\footnotesize $z_1$} & {\footnotesize $z_3$}  \\ \hline
{\small $b$} & {\footnotesize $z_2$} & {\footnotesize $z_4$}  \\ \hline
{\small $u$} & {\footnotesize $z_5$} & {\footnotesize $z_5$}  \\ \hline
\end{tabular}
};
\end{tikzpicture}
\caption{$\mathrm{rn}_Z(\widehat{G})$ of the game in Figure \ref{fig:noindependent}.\label{fig_BoS_dissipative}}
\end{figure} 


It follows from the algorithm described in Lemma \ref{lem:lemmaisomorphic} that player \textsl{1} is active at the root $\varnothing$ with three feasible actions, whereas player \textsl{2} is inactive at $\varnothing$ because the unique partition of $\{x,y\}$ which divides the terminal paths in disjoint sets is $\{x,y\}$ itself. At this point, for the sub-games starting at the histories $(a)$ and $(b)$, player \textsl{2} is active and the finest partition dividing the terminal paths in disjoint sets is $\{\{x\},\{y\}\}$. It follows that player \textsl{2} is active at such histories, they are on the same information set, and he has two available actions. Finally, player \textsl{2} is inactive at the history $(u)$. In other words, we constructed the game $\widehat{G}$ whose extensive game structure is represented in the right hand side of Figure \ref{fig:noindependent}.
\end{example}

We can thus provide a characterization of invariant transformations.
\begin{lemma}\label{thm:maininvariant}
A correspondence $T: \mathcal{G} \twoheadrightarrow \mathcal{G}$ with $\mathrm{dom}(T) \neq \emptyset$ is an invariant transformation if and only if it satisfies \eqref{eq:condition}.
\end{lemma}
\begin{proof}
By  
Remark \ref{remark:compositioninvarianttransformations}, $\mathcal{T}$ contains all the 
correspondences $T: \mathcal{G} \twoheadrightarrow \mathcal{G}$ with $\mathrm{dom}(T) \neq \emptyset$ which satisfy \eqref{eq:condition}. 
Conversely, fix an invariant transformation $T$ and an extensive game structure $G \in \mathrm{dom}(T)$ (which is possible, since $\mathrm{dom}(T)\neq \emptyset$). 
Then, for each $G^\prime \in T(G)$, we have that $G$ is behaviorally equivalent to $G^\prime$, i.e., $\mathrm{rn}_Z(G) \simeq \mathrm{rn}_Z(G^\prime)$. 
By Lemma \ref{lem:lemmaisomorphic}, there exist unique minimal game structures $\widehat{G}$ and $\widehat{G^\prime}$ with $Z$-reduced normal forms $\mathrm{rn}_Z(G)$ and $\mathrm{rn}_Z(G^\prime)$, respectively. 
It follows by Lemma \ref{lem:minimalextensive} that $\widehat{G}$ and $\widehat{G^\prime}$ are the minimal reductions of $G$ and $G^\prime$, respectively. 
Therefore there exist $n,k \in \mathbb{N}_0$ and $\varphi_1,\ldots,\varphi_n, \psi_1,\ldots,\psi_k \in \{\gamma,\sigma\}$ such that 
$$
\widehat{G} \in (\varphi_1 \circ \cdots \circ \varphi_n)(G) \,\,\,\,\text{ and }\,\,\,\,\widehat{G^\prime} \in (\psi_1 \circ \cdots \circ \psi_k)(G^\prime).
$$
This implies that $G^\prime \overset{\mathrm{iso}}{\in} (\psi_k^{-1} \circ \cdots \circ \psi_1^{-1} \circ \varphi_1 \circ \cdots \circ \varphi_n)(G)$, so that $T$ satisfies \eqref{eq:condition}.
\end{proof}

Finally, we are ready to prove our main result.
\begin{proof}[Proof of Theorem \ref{th:mainresult}]
\ref{item:a1} $\Longleftrightarrow$ \ref{item:a2}. The \textsc{if} part is clear. 
The \textsc{only if} part follows from Lemma \ref{lem:lemmaisomorphic} and Lemma \ref{lem:minimalextensive}.

\ref{item:a2} $\Longleftrightarrow$ \ref{item:a3}. The \textsc{only if} part is clear. 
The \textsc{if} part follows from Lemma \ref{lem:minimalextensive}.


\ref{item:a1} $\Longleftrightarrow$ \ref{item:a4}. The \textsc{if} part is clear, cf. Remark \ref{remark:compositioninvarianttransformations}. Conversely, if $G$ is behaviorally equivalent to $G^\prime$ then there exists an invariant transformation $T \in \mathcal{T}$ such that $G^\prime \in T(G)$. By Lemma \ref{thm:maininvariant}, $T: \mathcal{G}\twoheadrightarrow \mathcal{G}$ is a correspondence with $\mathrm{dom}(T)\neq \emptyset$ satisfying \eqref{eq:condition}. This means that there exist $n \in \mathbb{N}$ and $\iota_1,\ldots,\iota_n \in \{\gamma,\gamma^{-1},\sigma,\sigma^{-1}\}$ such that $G^\prime \overset{\mathrm{iso}}{\in} (\iota_1 \circ \cdots \circ \iota_n)(G)$.
This completes the proof.
\end{proof}

\subsection{Concluding remarks}\label{sec:conclusions}

The four transformations put forward by Thompson \cite{Thompson} are the
starting point for the development of many other equivalence relations. Some
solution concepts, like Nash equilibrium, iterated (weak or strict)
dominance, and strategic stability (Kohlberg and Mertens \cite{KohlbergMertens}) are
invariant to transformations that preserve the reduced normal form. Others,
like subgame perfection or trembling hand perfection in the agent normal
form (Selten \cite{Seltenspe, Seltenthp}), are only invariant to 
Interchanging/Simultanizing.  
The latter is almost a must if games are
represented in the traditional way of Kuhn \cite{MR0054924}, which does not allow for
a direct representation of simultaneity, but not if the formalism allows for
a direct representation of simultaneity, as in our case. It is worth noting
that some solution concepts, like weak perfect Bayesian
equilibrium (Myerson \cite{Myersonbook}, Mas Colell
et al. \cite{MasColell}) are not even invariant to Interchanging, as shown by
Myerson \cite[Figures 4.6 and 4.7]{Myersonbook}. 

We do not see different forms of invariance as strategic rationality
requirements, but rather as interesting properties that \emph{independently justified} solution concepts may or may not satisfy. In
particular, different notions of extensive-form rationalizability with an
independent epistemic foundation are based on the
notion of sequential best reply described in the Introduction, which applies
to classes of behaviorally equivalent strategies. With this, we look at
transformations that preserve behavioral equivalence and prove that two
simple invariant transformations, i.e., Interchanging/Simultanizing 
and Coalescing Moves/Sequential Agent Splitting, are sufficient to
characterize the notion of behavioral equivalence.%


\section{Acknowledgements}\label{sec:acknow}  
The authors are grateful to the editors Giacomo Bonanno and 
Wiebe can der Hoek and two anonymous reviewers for their careful reading and  thoughtful suggestions which improved the overall readibility of the article. The authors also thank Klaus Ritzberger (University of London) and Carlo Maria Cusumano, Nicodemo De Vito, Francesco Fabbri, David Ruiz Gomez, Paola Moscariello (Universit\`a Bocconi) for several useful comments. 
This research did not receive any specific grant from funding agencies in the public, commercial, or not-for-profit sectors.

\bibliographystyle{amsplain}
\bibliography{eefg}

\end{document}